\documentclass[paper=letterpaper, fontsize=12pt]{article}
\usepackage[letterpaper, total={7in, 9in}]{geometry}
\pdfoutput=1
\usepackage[T1]{fontenc}
\usepackage{fourier}

\usepackage{amsmath}
\usepackage{amsthm}
\usepackage{array}
\usepackage{float}
\usepackage[english]{babel}			
\usepackage[protrusion=true,expansion=true]{microtype}	
\usepackage{amsmath,amsfonts,amsthm} % Math packages
\usepackage{graphicx}	
\usepackage{relsize}
\usepackage{ifthen}
\usepackage{listings}
\usepackage{xcolor}
\usepackage{caption}
\usepackage{stmaryrd}
\usepackage{mathtools}
\usepackage{mdframed}
\usepackage{fancybox}

\makeatletter
\newenvironment{CenteredBox}{% 
	\begin{Sbox}}{% Save the content in a box
	\end{Sbox}\centerline{\parbox{\wd\@Sbox}{\TheSbox}}}% And output it centered
\makeatother

\mdfdefinestyle{MyFrame}{%
	linecolor=blue,
	outerlinewidth=2pt,
	linewidth=.5bp
	roundcorner=10pt,
	innertopmargin=\baselineskip,
	innerbottommargin=\baselineskip,
	innerrightmargin=10pt,
	innerleftmargin=10pt,
	backgroundcolor=gray!50!white}

\usepackage{url}

\usepackage[
backend=bibtex,
maxnames=99,
style=numeric-comp,
sorting=ynt
]{biblatex}
\usepackage{hyperref}

\theoremstyle{plain}
\newtheorem{THM}{Theorem} % reset theorem numbering for each chapter
\newtheorem{LMA}{Lemma} % reset theorem numbering for each chapter
\theoremstyle{definition}
\newtheorem{DEFN}[THM]{Definition} % definition numbers are dependent on theorem numbers

\numberwithin{equation}{section}		% Equationnumbering: section.eq#
\numberwithin{figure}{section}			% Figurenumbering: section.fig#
\numberwithin{table}{section}				% Tablenumbering: section.tab#

\newcommand{\MSD}[1][0]{
\ifthenelse{ \equal{#1}{0} }{
	\bf{msd}
}{
  	\bf{msd\_{#1}}
}}
\newcommand{\LSD}[1][0]{
\ifthenelse{ \equal{#1}{0} }{
	\bf{lsd}
}{
  	\bf{lsd\_{#1}}
}}

\newcommand{\ReverseBlack}{
{\textasciigrave}
}

\definecolor{cmd}{rgb}{0.0, 0.4, 0.0}
\definecolor{out}{rgb}{0.57, 0.64, 0.69}
\lstdefinestyle{pre}{
	language=bash,
	columns=fullflexible,
	showstringspaces=false,
	escapeinside={(*}{*)},
}
\lstdefinestyle{cmd}{
	language=C,
	stringstyle=\color{blue},
	keywordstyle=\color{cmd},     
	otherkeywords={exit,eval,def,macro,reg,load},
	emptylines=1,
	breaklines=true,
	basicstyle=\color{black},
	showstringspaces=false,
	escapeinside={(*}{*)},
	moredelim=**[is][\color{out}]{(@}{@)},
	moredelim=**[is][\color{red}]{(!}{!)},
}
\lstdefinestyle{reg}{
	language=C,
	stringstyle=\color{blue},
	keywordstyle=\color{cmd},     
	otherkeywords={exit,eval,def,reg,load},
	emptylines=1,
	breaklines=true,
	basicstyle=\color{black},
	showstringspaces=false,
	escapeinside={(@}{@)},
}
\lstdefinestyle{regerr}{
	language=C,
	deletekeywords={char},
	emptylines=1,
	breaklines=true,
	basicstyle=\color{red},
	showstringspaces=false,
	escapeinside={(@}{@)},
}
\lstdefinestyle{file} {
	frame=l,
	framesep=4.5mm,
	framexleftmargin=2.5mm,
	fillcolor=\color{gray!50!white},
	rulecolor=\color{blue},
	numberstyle=\normalfont\color{black},
	numbers=left,
}
\lstdefinestyle{err}{
	language=C,
	deletekeywords={char},
	emptylines=1,
	breaklines=true,
	basicstyle=\color{red},
	showstringspaces=false,
	escapeinside={(*}{*)},
}
\lstdefinestyle{out}{
	language=C,
	emptylines=1,
	breaklines=true,
	basicstyle=\color{out},
	showstringspaces=false,
	escapeinside={(*}{*)},
}

\DeclareMathOperator{\AND}{\&}

\DeclareMathOperator{\A}{A}

\newcommand{\horrule}[1]{\rule{\linewidth}{#1}} 	% Horizontal rule

\title{
		\usefont{OT1}{bch}{b}{n}
		\horrule{0.5pt} \\[0.4cm]
		\huge Automatic Theorem Proving in\\
		Walnut
		\horrule{2pt} \\[0.5cm]
}
\author{
		\normalfont 								\normalsize
        Hamoon Mousavi\\[-3pt]		\normalsize
        \today
}
\date{}

\begin{document}
\maketitle
\tableofcontents
\section{Introduction}\label{sec:introduction}
Walnut is a software package that implements a \emph{mechanical
decision procedure} for deciding certain combinatorial properties of
some special words referred to as \emph{automatic words} or
\emph{automatic sequences}. To learn more about automatic words and
their applications, see \cite{Allouche&Shallit:book}. To learn about
decision procedures for automatic words, see Schaeffer's Master's
thesis \cite{Schaeffer:thesis} and the survey paper
\cite{Shallit:survey}. To read more about decidable properties of
automatic words, refer to
\cite{Charlier&Rampersad&Shallit:enumeration}. To read about another
software package that provided a similar mechanical decision procedure
for automatic words, and was developed before Walnut, read Goc's
Master's thesis \cite{Goc:thesis}. To see applications of Walnut, refer
to \cite{paperfolding,tribonacci,fib1,fib2,fib3,balanced}.

The aim of this article is to introduce Walnut and explain its core
features. This article consists of four parts: basics, syntax,
implementation, and the Walnut guide. In the first part, Section
\ref{sec:basics}, we establish the basic notation and concepts. We go
over words, automata, number systems, automatic words, and Presburger
arithmetic. We learn what it means for an automaton to accept a
predicate. We also learn how to automatically decide properties of
automatic words.

The second part, Section \ref{sec:syntax and semantic}, talks about the
building blocks of predicates: constants, variables, operators, and
different types of expressions. The semantics of predicates in
Presburger arithmetic are well-known and are not explained, whereas
semantic rules for calling and indexing, with which we extend the
Presburger arithmetic to include automatic words, are explained in
detail.

The third part, Sections \ref{sec:decision procedure} and
\ref{sec:special automata}, explains the decision procedure implemented
in Walnut. The cross product of two automata, which is behind the
construction of automata for all binary logical operators, is
introduced. Building on that, we see how to construct automata for
predicates from automata for subpredicates. In Section \ref{sec:special
automata}, we talk about two types of automata that do not appear often
in Walnut, but are nevertheless important to understand.

The fourth and last part, Sections
\ref{sec:installation}--\ref{sec:working with input and output}, starts
with Walnut's installation and goes over all of its commands, i.e.,
exit, eval, def, reg, and load. In Section \ref{sec:working with input
and output}, we learn how to manually define automata in text files. We
also learn how to define new number systems.

If you are already familiar with the objects described in the first
sentence of this introduction, you can skip Section \ref{sec:basics}
and come back to it only as a reference.  For a more comprehensive
treatment of the theory behind decision procedures for automatic words
refer to
\cite{Schaeffer:thesis,Shallit:survey,Charlier&Rampersad&Shallit:enumeration}.

Since this article is more about Walnut than the theory behind it, when
we explain the latter, we use Walnut's notation as opposed to the more
familiar mathematical notation. For example, we use $\AND$ and $\A$ for
conjunction and universal quantifier as opposed to $\wedge$ and
$\forall$ of mathematical logic \footnote{Users enter logical
predicates in a terminal when they use Walnut. We find that entering
latex-like commands in the terminal, e.g., \textbackslash forall, does
not improve the readability.}. As another example, when we define
structures such as number systems or objects such as automatic words,
we give the definitions that are closer to Walnut's capabilities than
the most general theoretical ones possible. This will help the reader make
a smoother transition from the theory to its application in Walnut.

You can download Walnut from Jeffrey Shallit's
\href{https://cs.uwaterloo.ca/~shallit/papers.html}{website}, or alternatively from \href{https://github.com/hamoonmousavi/Walnut.git}{GitHub}. Walnut is
written in Java and is open source. It is licensed under GNU General
Public License. We would appreciate it if users cite this article in
their publications. For automata minimization Walnut relies on Valmari's minimization algorithm \cite{valmari:2012}. In order to use this minimization algorithm in Walnut, we manually translated Valmari's C++ implementation \cite{dfa.minimizer} almost varbatim to Java. For converting regular
expressions to automata, Walnut relies on the automata library in
\cite{dk.brics}. We would greatly appreciate it if users report bugs to
\href{mailto:sh2mousa@uwaterloo.ca}{sh2mousa@uwaterloo.ca}. The author would like to thank Jeffrey Shallit for revising this article. 

\section{Basics}\label{sec:basics}
\subsection{Words and Automata}\label{sec:words and automata}
A word $(a_i)_{i\in I}$ for a finite, infinite, or a possibly empty
subset $I$ of natural numbers $\mathbb{N}$, is a sequence of symbols
$a_i$ over a finite set called an alphabet. The set $I$ usually equals
$\mathbb{N}$ or $\mathbb{N}_l=\{k\in \mathbb{N}:k<l \}$ for some $l$.
The set of finite and infinite words over the alphabet $\Sigma$ are
denoted by $\Sigma^*$ and $\Sigma^\omega$, respectively. The empty word
is denoted by $\epsilon$. For the finite word $w = a_0a_1\cdots
a_{l-1}$, the length $|w|$, is defined and equals $l$. We let
$\Sigma^l$ denote the set of all words over $\Sigma$ of length $l$. A
subword (sometimes called ``factor'' in the literature) is a finite and
contiguous subsequence of a word. The subword of $w$ starting at
position $i$ of length $k\geq 0$ is denoted by $w[i..i+k-1]=a_i\cdots
a_{i+k-1}$. Many interesting properties of words can be expressed in
terms of their subwords. For example, the property of having two equal
and adjacent subwords, referred to as a square, is discussed in
numerous papers in the area of combinatorics on words. The product of
two words $x$ and $y$, denoted by $xy$, is the result of concatenating
$x$ by $y$.

There are cases where our words are defined over alphabets consisting
of tuples of symbols, so let us fix our notation regarding these words.
For a word $w$ over an alphabet
$\Sigma_1\times\Sigma_2\ldots\times\Sigma_n$, we let the projection map
$\pi_j(w)$ for $1\leq j\leq n$ denote the word over $\Sigma_j$,
obtained from $w$ by looking at the $j$'th coordinates, i.e., words
$\pi_j(w)$ are uniquely defined by
$$w=\prod\limits_{i=0}^{|w|-1}\big(\pi_1(w)[i],\pi_2(w)[i],\ldots,\pi_n(w)[i]\big).$$
For example, for $w= (0,1)(1,1)(0,0)$ over $\{0,1\}\times \{0,1\}$ we
have $\pi_1(w)=010$ and $\pi_2(w)=110$.

The reader is probably familiar with the notions of deterministic and
nondeterministic finite state automata. In Walnut, an automaton $M$
with $n$ inputs (input tapes), is an $(n+4)$-tuple
$\big(Q,q_0,F,\delta,\Sigma_1,\Sigma_2,\ldots,\Sigma_n\big)$, where $Q$
is the (finite) set of states, $q_0\in Q$ is the initial state,
$F\subseteq Q$ is the set of final states,
$\delta:Q\times\Sigma_1\times\Sigma_2\times\cdots\times\Sigma_n\rightarrow
Q$ is the transition function, and $\Sigma_i$ is the alphabet of the
$i$'th input (tape). The automaton's alphabet is defined to be the
cross product $\Sigma_1\times\Sigma_2\times\cdots\times\Sigma_n$, and
the notions of accepting a word $w$ or a language over this alphabet is
defined as usual. A nondeterministic automaton is defined similarly,
except that the transition function is defined by
$\delta:Q\times\Sigma_1\times\Sigma_2\times\cdots\times\Sigma_n\rightarrow
2^Q$. In Walnut and throughout this article, the $\Sigma_i$ are finite
subsets of integers $\mathbb{Z}$.

Two automata are equal (isomorphic) if their underlying graphs are
isomorphic. Two automata are equivalent if they accept the same
language. There exists a determinization algorithm that converts a
nondeterministic automaton to an equivalent deterministic automaton.
There exists a minimization algorithm that converts an automaton to an
equivalent automaton with the least number of states (which is unique
up to isomorphism). It is known that extending the automata model by
allowing multiple initial states (similar to how there can be multiple
final states) does not add to the model's expressiveness.

Next we extend the notion of accepting languages to relations, since
the latter is more natural in Walnut:

\begin{DEFN}[relations computed by automata]\label{relation}
The relation $R\subset
\Sigma_1^*\times\Sigma_2^*\times\ldots\times\Sigma_n^*$
computed/accepted by $M$ is defined by
$$R=\big\{\big(\pi_1(w),\pi_2(w),\ldots,\pi_n(w)\big):M \text{ accepts
} w\big\}.$$ Since for every word $w$, the words $\pi_i(w)$ are all of
the same length, the relation $R$ accepted by an automaton is consisted
of tuples of the words of the same length, i.e., we have
$$R\subseteq\bigcup\limits_{l\geq
0}\big(\Sigma_1^l\times\Sigma_2^l\times\cdots\times\Sigma_n^l\big)\subset
\Sigma_1^*\times\Sigma_2^*\times\ldots\times\Sigma_n^* .$$ 
\end{DEFN}

For example, the language accepted by the following automaton is $L=(0,0)^*(1,1)(0,0)(0,1)$, whereas the relation accepted is $R=\big\{(w_1,w_2):w_1\in0^*100, w_2\in0^*101,\text{ and } |w_1|=|w_2| \big\}$:
\begin{figure}[H]
\centering
\includegraphics[scale=.5]{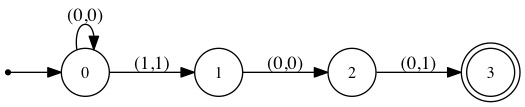}
\caption{Automaton accepting tuples of same length representations of $4$ and $5$ in binary}
\label{fig:automaton_accepting_tuples}
\end{figure}
In other words, the automaton accepts tuples $t=(w_1,w_2)\in
\{0,1\}^*\times\{0,1\}^*$ where $w_1$ and $w_2$ are representations of
the \textbf{the same length}, in the most-significant-digit-first
binary system, of natural numbers $4$ and $5$ respectively.  On the
other hand, referring to the words $w$ in $(\{0,1\}\times\{0,1\})^*$ that
are accepted by this automaton is not very descriptive. That is why, in
this article, we prefer the relation (tuple) terminology over the
language (word) terminology.

In almost all depictions of the underlying graphs of automata, such as the
one in Figure \ref{fig:automaton_accepting_tuples}, when a transition
is not specified, it is assumed to be a transition to a dead state. In
Walnut we do not store transitions to the dead state. Adding the dead
state and all implicit transitions to it, is called totalizing an
automaton.

An automaton with output is a tuple
$\big(Q,q_0,O,\delta,\Sigma,\Sigma_1,\Sigma_2,\ldots,\Sigma_n\big)$
where $Q, q_0,\delta, \Sigma_j$ are as before, the set $\Sigma$ is the
output alphabet, and, instead of a set of final states, we have a map
$O:Q\rightarrow \Sigma$. The symbol $O(q)$ is called the output of the
state $q$. An automaton with output can be thought of as an automaton
that reads a word over $\Sigma_1\times\Sigma_2\cdots\times\Sigma_n$ and
outputs whatever is the last state's output. In Walnut, the output
alphabet $\Sigma$ is a finite subset of integers. We can think of
ordinary automata as a special case of automata with output by letting
the set of final states to be $F=\{q:O(q)\neq 0\}$. This is indeed how
ordinary automata are stored in Walnut.

In the next section, we learn how to add more structure to alphabets by
defining number systems. As we saw in the example, the automaton in
Figure \ref{fig:automaton_accepting_tuples} accepts binary
representations of numbers. In a moment we will extend our definition
of automata to
$\big(Q,q_0,F,\delta,\mathbf{S_1},\mathbf{S_2},\ldots,\mathbf{S_n}\big)$,
where the $\mathbf{S_j}$ are number systems and concealed in them are
alphabets $\Sigma_{\mathbf{S_j}}$ among other things.

\subsection{Number Systems}\label{sec:number systems}
In any course on theory of computation, it is customary to talk about
the representations of the objects an algorithm/Turing machine takes as
inputs. At the core of Walnut are automata taking natural numbers as
inputs, and doing various computations on them, so fixing a
representation for natural numbers is essential. We could limit
ourselves to binary representations. However, there are many
interesting automata accepting representations in number systems other
than the binary one. So we are going to define, in general terms, the
concept of a number system. Walnut allows number systems to be defined
and used (with a few restrictions to the general definition below).

\begin{DEFN}[number systems]\label{def:number systems}
A number system $\mathbf{S}$ is a $3$-tuple $(\Sigma_\mathbf{S},R_\mathbf{S},[]_\mathbf{S})$ of alphabet $\Sigma_\mathbf{S} \supseteq \{0,1\}$, language $R_\mathbf{S} \subset \Sigma_\mathbf{S}^*$ of valid representations containing $0^*$ and at least one of $0^*1$ or $10^*$, and decoding function $[]_\mathbf{S}:R_\mathbf{S}\rightarrow \mathbb{N}$ that assigns integers to every word in $R_\mathbf{S}$ and for which $[]_\mathbf{S}(w)$ is usually written as $[w]_\mathbf{S}$. The decoding function has the following additional properties:
\begin{itemize}
\item $[z]_\mathbf{S}=0$ if and only if $z\in 0^*$
\item $[1]_\mathbf{S}=1$
\item For all $w\in R_\mathbf{S}$, either $zw\in R_\mathbf{S}$ and $[zw]_\mathbf{S}=[w]_S$ for all $z\in 0^*$, or $wz\in R_\mathbf{S}$ and $[wz]_\mathbf{S}=[w]$ for all $z\in 0^*$. The former is called an $\MSD$ number system and the latter is called an $\LSD$ number system\footnote{$\MSD$ and $\LSD$ are short for most-significant-digit-first and least-significant-digit-first, respectively. 
However, it should not be taken literally in this definition, as one
could define $\MSD$ number systems (in the sense defined here), with
no direct correspondence to the notion of most-significant-digit-first
representation.}.  \item For all positive $n\in
\mathbb{N}$, there exists $w\in R_\mathbf{S}$ for which
$[w]_\mathbf{S}=n$ and $w[0]\neq 0$ if $\mathbf{S}$ is $\MSD$ or
$w[|w|-1]\neq 0$ if $\mathbf{S}$ is $\LSD$. The word $w$, if unique, is called the
canonical encoding of $n$ in $\mathbf{S}$, and is sometimes denoted by
$(n)_\mathbf{S}$. We let $(0)_\mathbf{S}=\epsilon$.

\end{itemize}
The addition relation $+_\mathbf{S} \subset R_\mathbf{S}^3$ is defined such that $(x,y,z)\in +_\mathbf{S}$ if and only if $x,y,z$ are of the same length and $[x]_\mathbf{S}=[y]_\mathbf{S}+[z]_\mathbf{S}$. The equality relation $=_\mathbf{S}\subset R_\mathbf{S}^2$ is defined such that $(x,y)\in =_\mathbf{S}$ if and only if $x$ and $y$ are of the same length and $[x]_\mathbf{S}=[y]_\mathbf{S}$.  The less than relation is defined as $<_\mathbf{S} \subset R_\mathbf{S}^2$ for which $(x,y)\in <_\mathbf{S}$ if and only if $x$ and $y$ are of the same length and $[x]_\mathbf{S} < [y]_\mathbf{S}$. We adopt the in-order notation for $+_\mathbf{S}$, $=_\mathbf{S}$, and $<_\mathbf{S}$, i.e., we write $x=y+_\mathbf{S}z$, $x=_\mathbf{S}y$, and $x<_\mathbf{S}y$ as opposed to the more cumbersome $(x,y,z)\in +_\mathbf{S}$, $(x,y)\in =_\mathbf{S}$, and $(x,y) \in <_\mathbf{S}$ respectively. It follows from the definition that for all $n\in \mathbb{N}$, the set of representations of $n$ in $\mathbf{S}$, defined by $\{w:[w]_{\mathbf{S}}=n\}$ is non-empty.
\end{DEFN}

For example, the most-significant-digit binary system, denoted by $\MSD[2]$, is defined by $(\{0,1\},\{0,1\}^*,[]_{\MSD[2]})$ where $$[w]_{\MSD[2]} = \mathlarger{\sum\limits_{i=0}^{|w|-1}} [w[i]]_{\MSD[2]}2^{|w|-i-1},$$ e.g., $[001001]_{\MSD[2]} = 0\cdot 2^5 + 0\cdot 2^4 + 1\cdot 2^3+0\cdot 2^2+0\cdot 2^1+1\cdot 2^0=9$. For $\MSD[2]$, we are very fortunate to have simple automata computing all of its important aspects, namely, valid representations $R_{\MSD[2]}$, the addition relation $+_{\MSD[2]}$, the equality relation $=_{\MSD[2]}$, and the less-than relation $<_{\MSD[2]}$. See Figures \ref{fig:binary_representations},\ref{fig:binary_addition},\ref{fig:binary_equality}, and \ref{fig:binary_less_than} respectively. 
\begin{figure}[H]
\centering
\includegraphics[scale=.5]{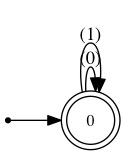}
\caption{Automaton computing $R_{\MSD[2]}$}
\label{fig:binary_representations}
\end{figure}

\begin{figure}[H]
\centering
\includegraphics[scale=.5]{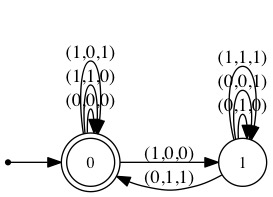}
\caption{Automaton computing $+_{\MSD[2]}$}
\label{fig:binary_addition}
\end{figure}

\begin{figure}[H]
	\centering
	\includegraphics[scale=.5]{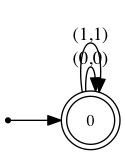}
	\caption{Automaton computing $=_{\MSD[2]}$}
	\label{fig:binary_equality}
\end{figure}

\begin{figure}[H]
\centering
\includegraphics[scale=.5]{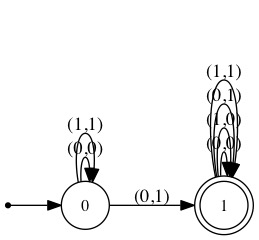}
\caption{Automaton computing $<_{\MSD[2]}$}
\label{fig:binary_less_than}
\end{figure}

We can define the least-significant-digit-first binary system, denoted
by $\LSD[2]$, in a similar way. In fact, we can define $\MSD[n]$ and
$\LSD[n]$ for all $n\geq 2$, and for all of them, there are simple
automata computing valid representations, addition, equality, and
less-than relations. In fact we can define the following:

\begin{DEFN}[number systems in Walnut]\label{def:walnut_number_system}
Number systems for which the automata for representations, addition,
equality, and less-than exist, and equality is the same as word
equality, i.e., $x=_\mathbf{S}y$ if and only if $x=y$, are exactly the
type of number systems one can define and use in Walnut. Note that the
alphabet of a number system is restricted to finite subsets of
$\mathbb{Z}$ due to the same restriction on automata in Walnut.
\end{DEFN}
In addition to base-$n$ number systems, Walnut has a
built-in definition for the Fibonacci number system.

The most-significant-digit-first Fibonacci system, denoted by $\MSD[fib]$, is defined by $(\{0,1\},0^*(\epsilon\mid 1)(0\mid 01)^*,[]_{\MSD[fib]})$ where
$$[w]_{\MSD[fib]} = \mathlarger{\sum\limits_{i=0}^{|w|-1}} [w[i]]_{\MSD[fib]}F_{|w|-i-1},$$
where $F_i$ is the $i$'th Fibonacci number given by $F_0=1,F_1=2$, and $F_i=F_{i-1}+F_{i-2}$ for $i\geq 2$. For example, $[001001]_{\MSD[fib]}= 0\cdot F_5 + 0\cdot F_4 + 1\cdot F_3+0\cdot F_2+0\cdot F_1+1\cdot F_0=6$. The set of valid representations is exactly the set of binary words avoiding consecutive $1$s. The avid reader might want to verify that $\MSD[fib]$ is a number system. There are automata computing all major aspects of $\MSD[fib]$. For example, here is the automaton accepting $R_{\MSD[fib]}$\footnote{The automaton accepting $+_{\MSD[fib]}$ has $16$ states, which is too big to be represented here.}:
\begin{figure}[H]
\centering
\includegraphics[scale=.5]{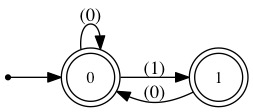}
\caption{Automaton computing $R_{\MSD[fib]}$}
\label{fig:fib_representations}
\end{figure}

In cases, where an automaton's inputs are representations of integers in
some number system, which by far are the most important type of
automata in Walnut, we would like to signify these number systems
instead of the input alphabets. For example, we might write
$\big(Q,q_0,F,\delta,\mathbf{S_1},\mathbf{S_2},\ldots,\mathbf{S_n}\big)$
to mean
$\big(Q,q_0,F,\delta,\Sigma_{\mathbf{S_1}},\Sigma_{\mathbf{S_2}},\ldots,\Sigma_{\mathbf{S_n}}\big)$.
It should be understood that in these cases, if for a word $w$ input
$\pi_j(w)$ is not a valid representation in $\mathbf{S_j}$, it does not
mean that the automaton's behavior is not defined for $w$. This just
means that $w$ is, by default, not going to get accepted. The behaviors
of both automata and automata with output that are taking
representations of numbers in some number systems as inputs are
defined for all words (even those not representing numbers in the given
number systems).

\subsection{Automatic Words}\label{sec:automatic words}
An automatic word $W=(a_i)_{i\geq 0}$ is a word in $\Sigma^\omega$ for
which there exists a number system $\mathbf{S}$ and an automaton with
output $M\big(Q,q_0,O,\delta,\Sigma,\mathbf{S}\big)$ for which reading
$x\in R_{\mathbf{S}}$ outputs $W[[x]_\mathbf{S}] = a_{[x]_\mathbf{S}
}$. In other words, for an automatic word, the symbol at position $i$
for all $i$ can be effectively computed by running an automaton with
output on any single representation of $i$ in a number system. As usual
we assume $\Sigma$ is a finite subset of $\mathbb{Z}$.

The word $T$ for which the symbol at position $i$, is the number of $1$s in any binary representation of $i$, modulo $2$, is called the Thue-Morse word. The Thue-Morse word is well-defined since all the infinitely many different binary representations of an integer have the same number of $1$'s. It is instantly clear that $T$ is an automatic word over $\MSD[2]$ if one notes the automaton with output in Figure \ref{fig:thue}.
\begin{figure}[h]
	\centering
	\includegraphics[scale=.6]{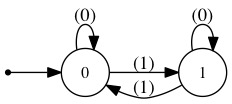}
	\caption{The Thue-Morse word}
	\label{fig:thue}
\end{figure}

In the introduction, we mentioned that Walnut decides some properties of automatic words. Recall from Section \ref{sec:words and automata} that squares are non-empty words of the form $xx$. It is easy to see that $T$ has square subwords. The following predicate captures this property:
$$\exists{i} \exists{n} \forall{j}, j<n \implies T[i+j]=T[i+n+j].$$

Walnut provides a decision procedure that takes predicates like this
and decides whether they are true or false. Walnut does so, by
constructing automata for every subpredicate in the predicate above;
see Section \ref{sec:automata accepting predicates} for more details.
It starts by constructing from the automaton in Figure \ref{fig:thue}
an automaton $M_1\big(Q,q_0,F,\delta,\MSD[2],\MSD[2],\MSD[2]\big)$ for
subpredicate $P_1\coloneqq T[i+j]=T[i+j+n]$. This means (see Section
\ref{sec:automata accepting predicates}) that $M_1$ is constructed so
that it accepts tuples $t=(w_1,w_2,w_3)$ if and only if
$|w_1|=|w_2|=|w_3|$ and substitutions $i=[w_1]_{\MSD[2]}$,
$j=[w_2]_{\MSD[2]}$, and $n=[w_3]_{\MSD[2]}$ are satisfying $P_1$.
Walnut then using $M_1$ constructs an automaton $M_2$ for
$P_2\coloneqq\forall{j}, j < n \implies T[i+j]=T[i+n+j]$. The automaton
$M_2$ takes two inputs representing the two free variables $i$ and $n$
in $P_2$. Walnut continues by constructing the automaton $M_3$ for
\linebreak$P_3\coloneqq\exists{n} \forall{j}, j < n \implies
T[i+j]=T[i+n+j]$. In the end, Walnut returns true if $M_3$ accepts
anything. The fact that $M_1$,$M_2$, and $M_3$ exist is explained in
Section \ref{sec:automata accepting predicates}. The details of how
Walnut constructs these automata are explained in Section
\ref{sec:decision procedure}. The details of what comprises a valid
predicate is explained in Section \ref{sec:syntax and semantic}. To see
more examples of the properties of the Thue-Morse word and their proofs
see Section \ref{sec:eval}.

We can extend the definition of automatic words to higher dimensions.
The ($n$-dimensional) automatic word
$$W=\big(a_{i_1,i_2,\ldots,i_n}\big)_{i_1\geq 0,i_2\geq 0,\ldots,
i_n\geq 0}$$ is an infinite word over $\Sigma$ for which there
exist number systems $\mathbf{S_j}$ and an automaton with output
$$M\big(Q,q_0,O,\delta,\Sigma,\mathbf{S_1},\mathbf{S_2},\ldots,\mathbf{S_n}\big)$$
for which reading $x$, such that $\pi_j(x) \in R_{\mathbf{S_j}}$ for
all $j$, outputs
$$W\big[[\pi_1[x]]_{\mathbf{S_1}}\big]\big[[\pi_2[x]]_{\mathbf{S_2}}\big]\cdots\big[[\pi_n[x]]_{\mathbf{S_n}}\big]
=
a_{[\pi_1[x]]_{\mathbf{S_1}},[\pi_2[x]]_{\mathbf{S_2}},\ldots,[\pi_n[x]]_{\mathbf{S_n}}}.$$

\subsection{Automata accepting Predicates}\label{sec:automata accepting predicates}
In Walnut, we are interested in automaton $M$ accepting same-length
representations in number systems
$\mathbf{S_1},\mathbf{S_2},\ldots,\mathbf{S_n}$ of integers
$x_1,x_2,\ldots,x_n$ satisfying some predicate $P$. When this is the
case we say that automaton $M$ accepts the predicate $P$ (or
equivalently$M$ accepts relation $R$ of tuples satisfying $P$). We
already saw a few examples of such automata in Figures
\ref{fig:automaton_accepting_tuples}--\ref{fig:fib_representations}.
From \cite{Buchi:1966}, also see \cite{Schaeffer:thesis}, and as it
will be proved again in Section \ref{sec:decision procedure}, for
predicate $P$ in Presburger arithmetic such an automaton always exists.
Presburger arithmetic is the first-order theory of natural numbers, in
which predicates are consisted of constants (natural numbers),
variables over natural numbers, existential quantifiers, universal
quantifiers, logical operators (conjunction, disjunction, negation,
exclusive disjunction, implication, equivalence), arithmetic operators
(addition, subtraction, multiplication and division by constants), and
comparison operators (equality, less than, greater than, less than or
equal, greater than or equal)\footnote{Presburger arithmetic in its
formal definition recognizes only a minimal subset of constants and
operators: $0$,$1$,$+$,$=$,$<$,$\forall$, but it is not difficult to
show that all the other objects and operators we mentioned, e.g.,
multiplication by constants, does not add to the power of Presburger
arithmetic and can be derived from that minimal set of objects. See
Section \ref{sec:arithmetic and alphabetic constants} for more details.
One thing to note here is that subtraction $a-b$ exists only when there
exists a non-negative number $c$ for which $b+c=a$.}.

You can find the list of all operators in table \ref{tab:operators}.
This list has three operators, namely, reverse\ReverseBlack,indexing
$[]$, and calling $\$$, that are not allowed in Presburger arithmetic.
By indexing we mean indexing into an automatic word, e.g., writing
things like $W[i+j]=W[i+n+j]$; see Section \ref{sec:indexing
expressions} for more details. In
\cite{Shallit:survey},\cite{Charlier&Rampersad&Shallit:enumeration},\cite{Schaeffer:thesis},
and also in Section \ref{sec:indexing to an automatic word} we learn
that extending Presburger arithmetic to include indexing is still
decidable. In Section \ref{sec:calling expressions} we learn about
calling and in Section \ref{sec:calling an automaton} we learn that it
is just a syntactic sugar and does not add to the power of the extended Presburger
arithmetic (one that includes indexing into automatic words). We learn about reverse operation in Section
\ref{sec:complement and reverse}. From here on, by ``predicate'' we mean a
predicate over this extended Presburger arithmetic (extended to include
indexing into automatic words) and until we see the proof in Section
\ref{sec:decision procedure}, we accept the fact that there exist
automata accepting such predicates.

In Section \ref{sec:syntax and semantic} we formally define what constitutes a predicate, but first let us see a few examples:
\begin{itemize}
\itemsep0em
\item $P_1\coloneqq a=4\mathbin{\&}b=5$
\item $P_2\coloneqq a=b+c$
\item $P_3\coloneqq\text{A}x\text{ }\text{E}y\text{ } x=2*y \mid x=2*y+1$
\item $P_4\coloneqq T[i+j]=T[i+n+j]$
\end{itemize}

We adopt the terminology of free variables from mathematical logic,
i.e., a variable that is not bound to a quantifier (quantified). For
example $P_3$ has no free variables, and can be regarded as a constant,
in this case it is always true.

We have seen that, given a predicate $P$, for any ordering
$x_1,x_2,\ldots,x_n$ of free variables and for every assignment of
number systems $\mathbf{S_1},\mathbf{S_2},\ldots,\mathbf{S_n}$ to those
variables, there exists an automaton $M$ accepting such a predicate,
i.e., a tuple of same length words $t=(w_1,w_2,\ldots,w_n)$ is accepted
by $M$ if and only if the substitutions $x_i=[w_i]_{\mathbf{S_i}}$ satisfy
$P$.

For example, consider the predicate $P_1$. The automaton in Figure
\ref{fig:automaton_accepting_tuples} accepts $P_1$. Furthermore there
exists automaton $M$ accepting tuples $(x,y)$ for which $|x|=|y|$ and
substitutions $a=[y]_{\MSD[2]}$, and $b=[x]_{\MSD[2]}$ are satisfying
$P_1$. There also exists an automaton $N$ accepting tuples $(x,y)$ for
which $|x|=|y|$ and substitutions $a=[x]_{\MSD[fib]}$ and
$b=[y]_{\LSD[2]}$ are satisfying $P_1$. By definition, both $M$ and $N$
also accept the predicate $P_1$.

We would like to annotate predicates so that they contain information
on number systems without ambiguity (we will see how shortly). For such
an annotated predicate $P$ and the ordering $x_1,x_2,\ldots,x_n$ on
free variables, there exists a unique minimized automaton accepting the
predicate. We denote this unique automaton by $$(x_1,x_2,\ldots,x_n):
P.$$

\emph{The ordering we fix on variables, in Walnut and throughout this
article, is the lexicographic ordering on the variables' name.}

The following are examples of annotated predicates\footnote{Names for variables, words, and automata in Walnut start with a letter and can contain alphanumerics and underscores. So to distinguish number system annotations in a predicate we use the prefix $?$.}:
\begin{itemize}
\itemsep0em
\item $P'_1\coloneqq\text{?msd\_2 } a=4\mathbin{\&}b=5$
\item $P'_2\coloneqq\text{?msd\_fib } a=b+c$
\end{itemize}
From the annotated predicate $P'_1$ we understand that $a,b,4,5$ should
all be interpreted in $\MSD[2]$ and $=$ should be interpreted as
$=_{\MSD[2]}$. Hence $(a,b): P'_1$ is the automaton accepting $\MSD[2]$
representations of $4$ and $5$ as its first and second inputs
respectively. Also from annotation ?msd\_fib in $P'_2$ it is clear what
to expect from automaton $(a,b,c):P'_2$.

We can annotate a predicate with multiple number systems, e.g., see
Figure \ref{fig:one_and_one}. Here are the rules with which we assign
number systems to constants, variables, and operators in a predicate:

\begin{itemize}
\item If ?S appears inside a pair of parentheses or brackets, then the number system $\mathbf{S}$ is effective from the place it occurs in the predicate to the nearest closing parenthesis or bracket\footnote{Brackets $[]$ only appear in indexing expressions. See Sections \ref{sec:syntax and semantic} and \ref{sec:indexing expressions} for more details.}.
\item If ?S appears outside all parentheses and brackets, then the number system $\mathbf{S}$ is effective from the place it occurs in the predicate to the end of predicate.
\item If none of the rules above applies, the number system is assumed to be $\MSD[2]$ by default.
\item\label{itm:number_system_declarations} It is assumed that the number systems do not contradict each other, i.e., a single variable cannot have two different number systems in one predicate, and all operands of an arithmetic or comparison operator must belong to the same number system. 
\end{itemize}

We saw in Figure \ref{fig:automaton_accepting_tuples}, the \emph{unique} automaton $(a,b):P'_1$. In Figure \ref{fig:fourAndThirteen}, we see the automaton $(a,b):a=4\mathbin{\&}b=13$ (recall that when the number system is not specified it is assumed to be $\MSD[2]$):
\begin{figure}[H]
	\centering
	\includegraphics[scale=.5]{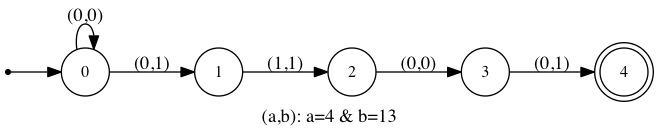}
	\caption{The automaton accepting $a=4\mathbin{\&}b=13$, does not accept all representations of $4$}
	\label{fig:fourAndThirteen}
\end{figure}

Note how this automaton fails to accept $t=(100,w_2)$ for any $w_2$.
This is obviously due to the fact that $13$ does not have a
representation of length $3$ in $\MSD[2]$. So we stress again that when
we say automaton $M$ accepts predicate $P$, we mean that $M$ accepts
all (tuples of) \emph{equal length representations} of $x_1,\ldots,x_n$
satisfying $P$. Therefore this example conforms to the definition.

Let us see an example of an automaton having multiple number systems. Figure \ref{fig:one_and_one} depicts the automaton \linebreak $(a,b): a=1\mathbin{\&}(\text{?lsd\_2 }b=1)$. 
\begin{figure}[h]
	\centering
	\includegraphics[scale=.5]{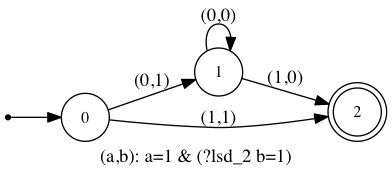}
	\caption{Automaton accepting $a=1\mathbin{\&}(\text{?lsd\_2 }b=1)$}
	\label{fig:one_and_one}
\end{figure}

\section{Syntax and Semantic of Predicates in Walnut}\label{sec:syntax and semantic}
\subsection{Alphabets}\label{sec:alphabets}
We mentioned in earlier sections that all input and output alphabets of
automata are subsets of integers in Walnut. Specifically for any
automatic word $W$, we can assume $W[i]$ is an integer.

\subsection{Arithmetic and Alphabetic Constants}\label{sec:arithmetic and alphabetic constants}
Arithmetic constants in a predicate are allowed to be natural
numbers only. There is, however, another type of constant:  the alphabetic
constant. Alphabetic constants are useful when referring to symbols at
particular positions in automatic words. For example, the predicate that
accepts positions for which the automatic word $W$ is $1$ is written as
$W[i]=@1$. In order to draw the distinction between alphabetic and
arithmetic constants, we use alphabetic constants with a prefix of $@$.
The reason we call these constants alphabetic (as opposed to
arithmetic) is due to the fact that Walnut does not allow (and it does
not make much sense to allow) predicates that are comparing indexing
expressions \ref{sec:indexing expressions} and arithmetic expressions
\ref{sec:arithmetic expressions}, e.g., expressions such as
$W[i]=a+b$ is not allowed. As
we will see in Section \ref{sec:relative expressions}, the only objects
that can be compared with indexing expressions are alphabetic constants
and indexing expressions themselves.

Alphabetic constants are ordered like ordinary integers, so we can
compare alphabetic constants, just like we can compare arithmetic
constants. For example, $@-1<@1$ is a valid predicate, and it is always
true; see Sections \ref{sec:indexing expressions} and \ref{sec:relative
expressions}. \emph{However, we cannot use alphabetic constants in arithmetic
expressions.}

\subsection{Variables}\label{sec:variables}
A variable's name must start with a letter and can contain
upper- and lower-case 
alphanumerics and underscores. A variable's name cannot be E or A.

\subsection{Operators}\label{sec:operators}
The full list of operators allowed in predicates can be found in Table
\ref{tab:operators}\footnote{we prefer this notation to those familiar
from mathematical logic, because we want to liken our notation to those
of programming languages, as Walnut is ultimately a programming
language.}. This list has operator precedences. The lower this number
is, the higher the precedence is. For example, multiplication by constant
has the highest precedence. Parentheses override all precedences. All
operators are associative from left to right, except for complement
$\sim$, reverse \ReverseBlack, quantifiers E and A, calling \$, and
indexing $[]$ which are all associative from right to left.

\begin{table}[h]
	\centering
	\begin{tabular}{ | m{1.7cm} | m{1.4cm}| m{4cm} | m{5cm}|} 
		\hline
		precedence & operator & explanation & examples\\
		\hline
		1 & $*$ & multiplication by a constant & $2*x$ and $x*2$ \\ 
		\hline
		1 & $/$ & division by a constant & $x/2$ but not $2/x$\\ 
		\hline
		2 & $+$& addition & \\ 
		\hline
		2 & $-$& subtraction & \\
		\hline
		3 &   $=$ & equality& \\
		\hline
		3 &  !=  & inequality& \\
		\hline
		3 &  $<$ & less than& \\
		\hline
		3 &  $>$  & greater than& \\
		\hline
		3 & <= & less than or equal& \\
		\hline
		3 &  >= & greater than or equal& \\
		\hline
		4 &  $\sim$  & complement&\\
		\hline
		4 &   \ReverseBlack  & reverse&\\
		\hline
		5 &  $\&$ & conjunction&\\
		\hline
		5 &  $\mid$  & disjunction&\\
		\hline
		5 &  $\wedge$ & exclusive disjunction&\\
		\hline
		6 &  => & implication&\\
		\hline
		7 & <=> & equivalence&\\
		\hline
		8 &   E  & existential quantifier & $\text{E}x,y,z$ or $\text{E}x\text{ E}y\text{ E}z$\\
		\hline
		8 &   A & universal quantifier & $\text{A}x,y,z$ or $\text{A}x\text{ A}y\text{ A}z$\\
		\hline
		9 & $\$$ & calling &$\$\text{M}(x,y)$ \\
		\hline
		9 & $[]$ & indexing & $T[i+j]$\\
		\hline
	\end{tabular}
	\caption{List of operators in Walnut}
	\label{tab:operators}
\end{table}

\subsection{Arithmetic Expressions}\label{sec:arithmetic expressions}
The permissible arithmetic operators are $+,-,*,/$. Equality $=$ is {\it not\/}
an arithmetic operator. A constant expression is an expression
involving only constants and arithmetic operators that evaluates to a
natural number, e.g., $4,3+2,6/4,2*3$ but not $-3$ nor $2-3$.  An
arithmetic expression is defined recursively in the usual way:

\begin{itemize}
\item A constant expression is an arithmetic expression, e.g., $2$,$10$,$7-4$, but not $-1$.
\item A variable is an arithmetic expression, e.g., $x,y,z$,etc.
\item For arithmetic expression $e$, the expression $(e)$ is also arithmetic.
\item For arithmetic expression $e_1$ and $e_2$ both of $e_1+e_2$ and $e_1-e_2$ are arithmetic expressions.
\item For variable $x$ and constant expression $c$ all of $x*c$,$c*x$, and $x/c$ are arithmetic expressions.
\item For arithmetic expression $e$ and constant expression $c$ all of $(e)*c$, $c*(e)$, and $(e)/c$ are arithmetic expressions. 
\end{itemize}
An arithmetic expression on its own is not a predicate, and it is not
meaningful to talk about an automaton accepting an arithmetic
expression. For example, talking about an automaton accepting $x+y+z=0$
makes sense, while talking about an automaton accepting $x+y+z$ is not
meaningful.  Walnut reports an error if the user tries to construct an
automaton for an arithmetic expression.

See Section \ref{sec:arithmetic and comparison} to see how Walnut
constructs automaton for valid predicates like $$(y_1 \otimes y_2
\otimes \cdots \otimes y_m) \olessthan (x_1 \otimes x_2 \otimes \cdots
\otimes x_n),$$ where the $x_i$ and $y_j$ are variables or arithmetic
constants, $\otimes$ are arithmetic operators, and $\olessthan$ is a
comparison operator.

\subsection{Indexing Expressions and Their Semantic Rules}\label{sec:indexing expressions}
For an $n$-dimensional automatic word $W$, an indexing expression is
$W[e_1][e_2]\cdots[e_n]$ where the $e_i$ are either arithmetic
expressions or predicates with one free variable.

An indexing expression on its own is not a valid predicate, and it is
not meaningful to talk about automata accepting indexing expressions.
Smallest predicates involving indexing expressions are defined in
Section \ref{sec:relative expressions} and they involve comparison
operators.

We use indexing expressions to refer to positions indicated by $e_i$.
The semantic of predicates involving indexing expressions can be
derived from the following rule:

\begin{DEFN}[semantic rule regarding indexing]
Suppose automatic word $W$, expressions $e_1,e_2,\ldots,e_n$ where 
the $e_i$ are either arithmetic expressions or predicates with one free variable, free variables $x_1,x_2,\ldots,x_m$ occurring in the $e_i$, and an alphabetic constant $\alpha$ are given. Predicate $W[e_1][e_2]\cdots[e_n]=@\alpha$ is satisfied by substitutions $x_k=v_k$ for all $k$, if all of the following hold:
\begin{itemize}
\item If $e_i$ is an arithmetic expression, then $a_i$ is the value of the
$e_i$ when evaluated at $x_k=v_k$ for all $k$.
\item If $e_i$ is a predicate with one free variable, then it is satisfied by substitutions $x_k=v_k$ for all $k$. Let $a_i$ equals $v_k$ when $x_k$ is the free variable in $e_i$.
\item The symbol $W[a_1][a_2]\ldots[a_n]$ equals $\alpha$.
\end{itemize}
\end{DEFN}

Having this rule, coming up with similar rules for other comparison
operators, e.g., $W[e_1][e_2]\cdots[e_n]<@\alpha$, and even predicates
involving comparison of two automatic words, e.g.,
$W_1[e_1][e_2]\cdots[e_m]\text{>=}W_2[e'_1][e'_2]\cdots[e'_n]$, should
be straightforward. Recall that alphabetic constants are ordered just
like integers.

\subsection{Calling Expressions and Their Semantic Rules}\label{sec:calling expressions}
For an automaton $M$ with $n$ inputs a calling expression is
$\$M(e_1,e_2,\ldots,e_n)$ where the $e_i$ are either arithmetic
expressions or predicates with one free variable. For such an
expression, we say that $M$ is called with arguments
$e_1,e_2,\ldots,e_n$. A calling expression on its own is a valid
predicate, as we will see in Section \ref{sec:relative expressions}.

\begin{DEFN}[semantic rule regarding calling]
Suppose $M$ is the automaton $y_1,y_2,\ldots,y_n: P$ for some predicate
$P$. Suppose expressions $e_1,e_2,\ldots,e_n$ where the $e_i$ are
either arithmetic expressions or predicates with one free variable, and
free variables $x_1,x_2,\ldots,x_m$ occurring in the $e_i$ are given.
Predicate $\$M(e_1,e_2,\cdots,e_n)$ is
satisfied by substitutions $x_k=v_k$ for all $k$, if all of the
following hold:
	\begin{itemize}
		\item If $e_i$ is an arithmetic expression, then $a_i$ is the
		value of $e_i$ when evaluated at $x_k=v_k$ for all $k$.
		\item If $e_i$ is a predicate with one free variable, then it is satisfied by substitutions $x_k=v_k$ for all $k$. Let $a_i$ equals $v_k$ when $x_k$ is the free variable in $e_i$.
		\item $P$ is satisfied by substitutions $y_i=a_i$ for all $i$. 
	\end{itemize}
\end{DEFN}

\subsection{Relative Expressions}\label{sec:relative expressions}
Comparison operators are $=$,!=,$<$,$>$,<=, and >=. A relative expression is any of the following:
\begin{itemize}
\item An expression $e_1 \olessthan e_2$ where $e_1$ and $e_2$ are arithmetic expressions and $\olessthan$ is any comparison operator.
\item An expression $e_1 \olessthan e_2$ where $e_1$ and $e_2$ are indexing expressions and/or alphabetic constants and $\olessthan$ is any comparison operator.
\item A calling expression is a relative expression.
\end{itemize}
We stress that $W[a]=b+2$ is not a relative expression based on the definition above, since $W[a]$ is an indexing expression and $b+2$ is an arithmetic expression. 
We will see shortly that any relative expression is a predicate.
Section \ref{sec:arithmetic and comparison} explains how to construct
automata accepting relative expressions.

\subsection{Predicates}\label{sec:predicates}
A predicate is an expression formed from relative expressions and logical operators:
\begin{itemize}
\item Every relative expression is a predicate.
\item For every predicate $P$ all of $(P)$, ${\sim (P)}$ and ${\text{\textasciigrave}(P)}$ are predicates. 
\item For every predicate $P_1$ and $P_2$ all of $P_1 \mathbin{\&} P_2$, $P_1 \mid P_2$, $P_1 \wedge P_2$, $P_1 \mathbin{\text{ => }} P_2$, $P_1 \mathbin{\text{ <=> }} P_2$ are predicates.
\item For every predicate $P$ and free variables $x_1,x_2,\ldots,x_n$ both of $\text{E}x_1,x_2,\ldots,x_n\text{ }P$ and $\text{A}x_1,x_2,\ldots,x_n\text{ }P$ are predicates.
\end{itemize}
The semantic rules with which we assign true and false values to
predicates defined here can be obtained by adding the semantic rules
for indexing and calling to the well-known semantics of first-order
logic and Presburger arithmetic.

Walnut provides two commands for converting predicates to automata
accepting them: eval and def; see Sections \ref{sec:eval} and
\ref{sec:def}, respectively.

\section{Decision Procedure: Walnut's Implementation}\label{sec:decision procedure}
In this section, we learn about a procedure that takes a predicate and
constructs an automaton accepting that predicate. The procedure
explained here is what implemented in Walnut, and we shall call it the
decision procedure.

For every defined number system, Walnut knows the automata for valid
representations, addition, equality, and less-than
predicates/relations. Every predicate is ultimately built out of these
four predicates using logical operators. So we only need to explain the
construction of automata for complex predicates from automata for
simpler subpredicates. We start by explaining cross product in Section
\ref{sec:cross product}, which is the core object when constructing
automata for predicates formed from binary logical operators, i.e.,
$\&$,$\mid$,$\wedge$,=>,<=>. Then we move on to quantification in
Section \ref{sec:quantification}, explaining the construction of
automata for predicates formed from E and A operators. In Section
\ref{sec:complement and reverse}, we discuss construction of automata
for the complement $\sim$ and reverse \ReverseBlack operators. With
these tools at our disposal, we are on the right track to construct
automata for complex predicates formed from comparison and arithmetic
operators, e.g., $*$,$/$,$>$,<=,etc. which we explain in Section
\ref{sec:arithmetic and comparison}.

\subsection{Cross Product}\label{sec:cross product}
Let $M\big(Q,q_0,F,\delta,\mathbf{S_1},\ldots,\mathbf{S_m}\big)$ and
$M'\big(Q',q'_0,F',\delta',\mathbf{S'_1},\ldots,\mathbf{S'_n}\big)$ be
the automaton $(x_1,\ldots,x_m):P$ and $(x'_1,\ldots,x'_n):P'$
respectively. Let us assume that if $x_i=x'_j$ then
$\mathbf{S_i}=\mathbf{S'_j}$. Let $\{x''_1,\ldots,x''_p\}$ where $p\leq
m+n$ be the union of $\{x_1,\ldots,x_m\}$ and $\{x'_1,\ldots,x'_n\}$
and further assume that the $x''_i$ are appearing in lexicographic
order. Depending on whether $x''_k=x_i$ or $x''_k=x'_j$, let
$\mathbf{S''_k}$ denote $\mathbf{S_i}$ or $\mathbf{S'_j}$ respectively.
Then the cross product of $M$ and $M'$ denoted by $M\times M'$ is the
tuple $$\big(Q\times
Q',(q_0,q'_0),\delta'',\mathbf{S''_1},\ldots,\mathbf{S''_p}\big)$$
where the transition function is defined to be
$$\delta''\big((q,q'),(\gamma_1,\ldots,\gamma_p)\big)=\big(\delta(q,(\alpha_1,\ldots,\alpha_m)),\delta'(q',(\beta_1,\ldots,\beta_n))\big)$$
for $\gamma_k$ equals $\alpha_i$ or $\beta_j$ depending on whether
$x''_k=x_i$ or $x''_k=x'_j$ respectively. Note that $M\times M'$ is not
an automaton since a set of final states is not specified. For
$F''\subseteq Q\times Q'$, let $(M\times M')(F)$ denote the automaton
$\big(Q\times
Q',(q_0,q'_0),F'',\delta'',\mathbf{S''_1},\ldots,\mathbf{S''_p}\big)$.

\begin{THM}\label{thm:cross product}
For $F'' = \big\{(q,q'):q\in F \text{ and } q'\in F'\big\}$, the
automaton $(M\times M')(F'')$ accepts predicate $P\mathbin{\&}P'$.
Furthermore, minimizing $(M\times M')(F'')$, we obtain automaton
$(x''_1,\ldots,x''_p):P\mathbin{\&}P'$.
\end{THM}

\begin{proof}
Based on the definition for cross product, for $M\times M'$ to be
defined, the same variables in $P$ and $P'$ have to have the same
number systems assigned in $P$ and $P'$. But that is exactly the same
condition that needs to hold for number system annotations in
$P\mathbin{\&}P'$ to be consistent (in the sense defined in the last
bullet in Page \pageref{itm:number_system_declarations}).

Let $t=(w_1,\ldots,w_m)$ and $t'=(w'_1,\ldots,w'_n)$ such that $w_i \in
\Sigma_{\mathbf{S_i}}^*$ and $w'_j\in\Sigma_{\mathbf{S'_j}}^*$ where
$|w_i|$ and $|w'_j|$ are all equal and $w_i=w'_j$ whenever $x_i=x'_j$.
Let $t''=(w''_1,\ldots,w''_p)$ such that $w''_k=w_i$ or $w''_k=w'_j$
depending on whether $x''_k=x_i$ or $x''_k=x'_j$.

We have the following equivalent statements:
\begin{enumerate}
\item $(M\times M')(F'')$ accepts $t''$.
\item There is a path from $(q_0,q'_0)$ to $(q,q')\in F''$ in $M\times M'$ reading $t''$.
\item There is a path from $q_0$ to $q$ in $M$ reading $t$, and there is a path from $q'_0$ to $q'$ in $M'$ reading $t'$.
\item $M$ accepts $t$ and $M'$ accepts $t'$.
\item $P$ is satisfied by substituting $x_i=[w_i]_{\mathbf{S_i}}$ for all $i$, and $P'$ is satisfied by substituting $x'_j=[w'_j]_{\mathbf{S'_j}}$ for all $j$.
\item $P\mathbin{\&}P'$ is satisfied by substituting $x''_k=[w''_k]_{\mathbf{S''_k}}$.
\end{enumerate}
\end{proof}

Obviously both the construction of cross product and minimizing
automata can be carried out using algorithmic procedures. Therefore
Theorem \ref{thm:cross product} gives us a procedure for constructing
the automaton for conjunction.

With proper definitions for $F''$, we have similar theorems for
$P\otimes P'$ when $\otimes$ is any other binary logical operator.

Let us construct the automaton $(a,b):a=1\mathbin{\&}b=2$ from $(a):a=1$ in
Figure \ref{fig:one} and $(b):b=2$ in Figure \ref{fig:two}.

\begin{figure}[h]
\centering
\begin{minipage}[t]{.5\textwidth}
 	\centering
 	\includegraphics[scale=.5]{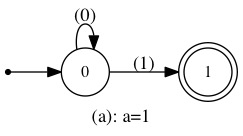}
 	\caption{Automaton $(a):a=1$}
 	\label{fig:one}
\end{minipage}%
\begin{minipage}[t]{.5\textwidth}
 	\centering
 	\includegraphics[scale=.5]{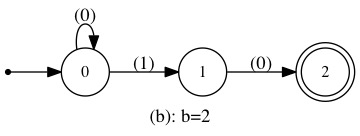}
 	\caption{Automaton $(b):b=2$}
 	\label{fig:two}
\end{minipage}%
\end{figure}
 Recall that transitions not depicted are transitions to a dead state. The
cross product operation is depicted below:
 \begin{figure}[H]
 	\centering
 	\includegraphics[scale=.5]{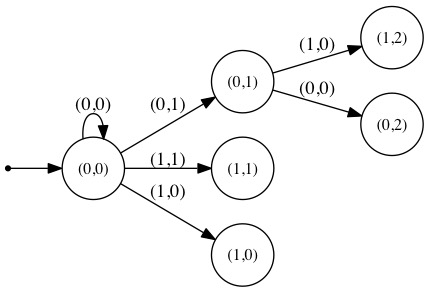}
 	\caption{Cross product $(a):a=1\times(b):b=2$}
 	\label{fig:cross_product}
 \end{figure}
 
Making $(1,2)$ a final state, minimizing, and renaming the states, we get the automaton in Figure \ref{fig:minimized}.
 \begin{figure}[H]
 	\centering
 	\includegraphics[scale=.5]{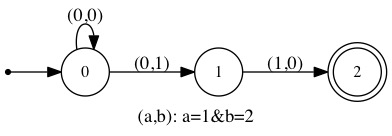}
 	\caption{Automaton $(a,b):a=1\mathbin{\&}b=2$}
 	\label{fig:minimized}
 \end{figure}

\subsection{Quantification}\label{sec:quantification}
In this section we learn how to construct an automaton
$(x_1,\ldots,x_{i-1},x_{i+1},\ldots,x_m): \text{E}x_{i}\text{ }P$ from
automaton $(x_1,\ldots,x_m): P$. Let
$M\big(Q,q_0,F,\delta,\mathbf{S_1},\ldots,\mathbf{S_m}\big)$ be the
automaton $(x_1,\ldots,x_m):P$ and let $P'$ be the predicate
$\text{E}x_{i}\text{ }P$. We first construct the nondeterministic
automaton $E(M,i)$
$$\big(Q,q_0,F,\delta',\mathbf{S_1},\ldots,\mathbf{S_{i-1}},\mathbf{S_{i+1}},\ldots,\mathbf{S_m}\big)$$
from $M$ by eliminating the $i$'th input (coordinate) on all
transitions, i.e., letting
$$\delta'\big(q,(\alpha_1,\ldots,\alpha_{i-1},\alpha_{i+1},\ldots,\alpha_m)\big)
=
\big\{\delta(q,(\alpha_1,\ldots,\alpha_{i-1},\alpha_i,\alpha_{i+1},\ldots,\alpha_m)):\text{
for all } \alpha_i \in \Sigma_{S_{i}} \big\}.$$

For example, letting $M$ be the automaton $(a,b):a=1\mathbin{\&}b=2$ depicted in Figure \ref{fig:minimized}, the automaton $E(M,2)$ is depicted as follows: 
\begin{figure}[H]
	\centering
	\includegraphics[scale=.5]{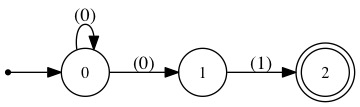}
	\caption{Non-deterministic automaton $E(M,2)$}
	\label{fig:quantified}
\end{figure}

By the definition of transition function of $E(M,i)$, i.e., $\delta'$,
it is easy to see that if $M$ accepts
\allowbreak$(w_1,\ldots,w_{i-1},w_i,w_{i+1},\ldots,w_m)$, then $E(M,i)$
accepts $$(w_1,\ldots,w_{i-1},w_{i+1},\ldots,w_m).$$ However, there
might be $t=(w_1,\ldots,w_{i-1},w_{i+1},\ldots,w_m)$ where the $|w_j|$
are equal for all $j\neq i$ and substitutions $x_j =
[w_j]_{\mathbf{S_j}}$ for all $j\neq i$, satisfies $P'$ but $E(M,i)$
does not accept $t$. In other words, there are cases where $E(M,i)$
does not accept $P'$.

In our example $M$ accepts $(0^n1,0^{n-1}10)$ for all $n\geq 1$, and as
it is clear $E(M,2)$ accepts $(0^n1)$ for all $n\geq 1$. However
$E(M,2)$ does not accept $(1)$, whereas $(1)$ should be accepted by any
automaton accepting $\text{E}b\text{ }(a=1\mathbin{\&}b=2)$.

Therefore, we have to do more work on $E(M,i)$, to get to an automaton
for $P'$. However as we will see in Lemma \ref{lemma:leading zeros
problem}, the automaton $E(M,i)$ might only miss an insignificant
portion of accepted tuples of an automaton accepting $P'$. These
insignificant tuples missed by $E(M,i)$ are those with leading or
trailing zeros. The good news is that with a little bit of technical
work, it is possible to revive even these insignificant tuples.

\begin{LMA}\label{lemma:leading zeros problem} 
Let $M$,$P$,$P'$, and $i$ be as in the discussion above, and suppose
$t=(w_1,\ldots,w_{i-1},w_{i+1},\ldots,w_m)$ is some tuple of same
length words. If $P'$ is satisfied with substitutions
$x_j=[w_j]_{\mathbf{S_j}}$ for $j\neq i$, then there exists a constant
$k\geq 0$ and\linebreak
$t_k=(w_{k,1},\ldots,w_{k,i-1},w_{k,i+1},\ldots,w_{k,m})$ such that for
all $j\neq i$ we have $w_{k,j}=0^kw_j$ or $w_{k,j}=w_j0^k$ depending on
whether $\mathbf{S_j}$ is $\MSD$ or $\LSD$, and $t_k$ is accepted by
$E(M,i)$. It is also the case that whenever $t_k$ for any $k\geq 0$,
with the appropriate substitutions, is satisfying $P'$, then $t$ is
also satisfying $P'$.
\end{LMA}

\begin{proof}
Substitutions $x_j=[w_j]_{\mathbf{S_j}}$ for $j\neq i$ satisfying
predicate $P'\coloneqq\text{E}x_i\text{ }P$ means that there exists a
natural number $v$, such that the substitutions above together with
$x_i=v$ is satisfying the predicate $P$.  By definition of number
systems, there exists a word $w_i$ such that $v=[w_i]_{\mathbf{S_i}}$.
Also by definition of number systems for any integer $y$ and word $w$,
if we have $y=[w]_{\mathbf{S}}$, then either $y=[0^kw]_{\mathbf{S}}$
for all $k\geq 0$ or $y=[w0^k]_{\mathbf{S}}$ for all $k\geq 0$
depending on whether $\mathbf{S}$ is $\MSD$ or $\LSD$.  Therefore there
exists an integer $k$ such that $w_{k,i}$ is either $0^kw_i$ or
$w_i0^k$ depending on whether $\mathbf{S_j}$ is $\MSD$ or $\LSD$ and
$v=[w_{k,i}]_{\mathbf{S_i}}$ and $|w_{k,i}|=|w_j|+k$ for all $j\neq i$.
Therefore $(w_{k,1},\ldots,w_{k,i-1},w_{k,i},w_{k,i+1},\ldots,w_{k,m})$
is accepted by $M$ where for all $j$ we have $w_{k,j}=0^kw_j$ or
$w_{k,j}=w_j0^k$ depending on whether $\mathbf{S_j}$ is $\MSD$ or
$\LSD$. Now by definition of $E(M,i)$, we know that
$t_k=(w_{k,1},\ldots,w_{k,i-1},w_{k,i+1},\ldots,w_{k,m})$ is accepted
by $E(M,i)$. This completes the first part of the lemma.

The second part follows very easily from the same properties of number
systems mentioned in the proof of the first part of the lemma.
\end{proof}

Based on Lemma \ref{lemma:leading zeros problem}, to get
$(x_1,\ldots,x_{i-1},x_{i+1},\ldots,x_m): \text{E}x_{i}\text{ }P$ we
just have to construct an automaton from $E(M,i)$ such that whenever
$t_k$ for any $k\geq 0$ is accepted, $t$ is also accepted. For the case
where all $\mathbf{S_j}$ for $j\neq i$ are either all $\MSD$ or all
$\LSD$, we can come up with an easy algorithm to revive tuples $t$ from
$t_k$. In case of all $\MSD$, let $I$ be the set of all states in
$E(M,i)$ reachable from the initial state by reading $(0,\ldots,0)^*$,
or in case of all $\LSD$, let $F'$ be the set of all states reaching to
a final state by reading $(0,\ldots,0)^*$.  We can compute $I$ or $F'$
using breadth-first search. In the case of $\MSD$ the nondeterministic
automaton\footnote{This is an automaton with multiple initial states.
One can show that for every nondeterministic automaton with multiple
initial states, there is an equivalent automaton with only one initial
state.}
 $$(Q,I,F,\delta',\mathbf{S_1},\ldots,\mathbf{S_{i-1}},\mathbf{S_{i+1}},\ldots,\mathbf{S_n})$$
and in the case of $\LSD$ the nondeterministic automaton
$$(Q,q_0,F',\delta',\mathbf{S_1},\ldots,\mathbf{S_{i-1}},\mathbf{S_{i+1}},\ldots,\mathbf{S_n})$$
is equivalent to $(x_1,\ldots,x_{i-1},x_{i+1},\ldots,x_m):
\text{E}x_{i}\text{ }P$. Determinizing and minimizing this automaton
gives us $(x_1,\ldots,x_{i-1},x_{i+1},\ldots,x_m): \text{E}x_{i}\text{
}P$.

In Figure \ref{fig:quantified}, the variable $a$ is over $\MSD[2]$. So
the set $I$ is $\{0,1\}$, therefore the following nondeterministic
automaton accepts $\text{E}b\text{ } a=1\mathbin{\&}b=2$:

\begin{figure}[H]
\centering
\includegraphics[scale=.5]{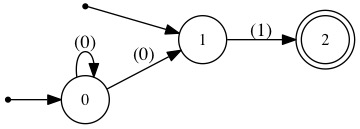}
\caption{Non-deterministic automaton accepting $\text{E}b\text{ } a=1\mathbin{\&}b=2$}
\label{fig:quantified2}
\end{figure}

Now determinizing and minimizing this automaton gives us $(a):\text{E}b\text{ } a=1\mathbin{\&}b=2$:
\begin{figure}[H]
	\centering
	\includegraphics[scale=.5]{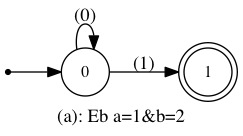}
	\caption{$(a):\text{E}b\text{ } a=1\mathbin{\&}b=2$}
	\label{fig:quantified3}
\end{figure}

\begin{mdframed}[style=MyFrame]
Currently if for $(x_1,\ldots,x_n): P$ it is not the case that for all
$j\neq i$ number systems $\mathbf{S_j}$ are all $\MSD$ or all $\LSD$,
then Walnut only constructs $E(M,i)$ for $\text{E}x_{i}\text{ }P$,
which is not theoretically accurate. So the user has to be very
cautious when quantifying predicates over mixed $\MSD$ and $\LSD$
number systems, or in cases where the quantified automaton is
non-arithmetic. For a definition of the latter see Section
\ref{sec:non-arithmetic automata}.
\end{mdframed}

To obtain an automaton for $Ax_i P$, note its equivalence to ${\sim
(\text{E}x_i\text{ }{\sim (P)})}$, where $\sim$ is the logical
complement (negation). See the next section to learn about the
complement operator.

\subsection{Complement and Reverse}\label{sec:complement and reverse}
To obtain $(x_1,\ldots,x_m): \sim (P)$ from $(x_1,\ldots,x_m): P$, one
has to add all transitions to dead state (in Walnut, we call this
totalizing an automaton), and then switching final and non-final
states, but one also has to make sure that the resulting automaton is
intersected with the automaton accepting
$R_{\mathbf{S_1}}\times\cdots\times R_{\mathbf{S_m}}$ where
$\mathbf{S_i}$ is the number system assigned to $x_i$ in (annotated)
predicate $P$. (Recall that $R_{\mathbf{S}}$ is the set of all valid
representations in the number system $\mathbf{S}$. Also recall that to
define and use a number system in Walnut, one has to provide automaton
accepting the set of all representations in that number system,
therefore automaton accepting $R_{\mathbf{S_1}}\times\cdots\times
R_{\mathbf{S_m}}$ could be constructed easily using cross product
explained in Section \ref{sec:cross product}.)

Take a look at automaton $(a): \text{?msd\_fib }a=1$ depicted in Figure \ref{fig:orig} that accept words representing $1$ in $\MSD[fib]$.
\begin{figure}[h]
	\centering
	\includegraphics[scale=.5]{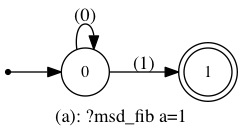}
	\caption{Number 1 in Fibonacci}
	\label{fig:orig}
\end{figure}

Now to obtain $(a):\sim(\text{?msd\_fib }a=1)$, we first add the dead state and all the transitions to it:
\begin{figure}[H]
	\centering
	\includegraphics[scale=.5]{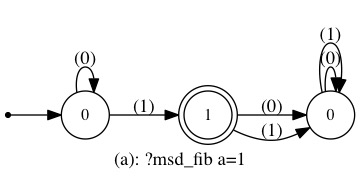}
	\caption{Totalized automaton}
	\label{fig:orig_dead_state}
\end{figure}
Switching final and non-final states we obtain an automaton accepting $\{0,1\}^*\setminus 0^*1$:
\begin{figure}[H]
	\centering
	\includegraphics[scale=.5]{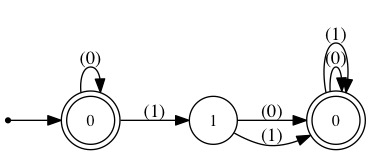}
	\caption{Switching final and non-final states}
	\label{fig:orig_switch_final}
\end{figure}
But then this automaton accepts words that have consecutive $1$'s which
are not acceptable Fibonacci representations. So to get the final
answer we have to intersect this automaton with the one depicted in
Figure \ref{fig:fib_representations}. The result is depicted in
\ref{fig:orig_complement}.
\begin{figure}[h]
	\centering
	\includegraphics[scale=.5]{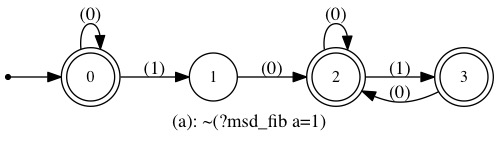}
	\caption{Automata accepting all numbers in Fibonacci except $1$}
	\label{fig:orig_complement}
\end{figure}

The reverse operator is not a logical operator per se, but we include
it because it is useful when working with automata. The operand of the
reverse operator is an automaton \footnote{Unlike the reverse operator,
operands for logical operators are predicates. The reader however
understands the very thin and superficial distinction between automata
and predicates in this article.}. The result is an automaton with all
its transitions reversed.

\subsection{Arithmetic and Comparison Operators}\label{sec:arithmetic and comparison}
Recall that for every number system $\mathbf{S}$ that we use in Walnut
the three automata $(a,b,c):\text{?S }a=b+c$, $(a,b): \text{?S
}a=b$\footnote{This automaton does not need to be defined explicitly by
the user, because we assumed for all number systems $\mathbf{S}$ in
Walnut $a=_{\mathbf{S}}b$ if and only if $a=b$.}, and $(a,b):\text{?S
}a<b$ are defined. In this section we show that using these three
automata and the decision procedure outlined in Sections \ref{sec:cross
product}--\ref{sec:complement and reverse}, we can construct automata
for more complex relative expressions with lots of arithmetic
operators.

For a constant $c>0$, a natural number, automata $(a):\text{?S }a=c$
can be constructed recursively using automata $(b):\text{?S }b=c'$ and
$(a,b):\text{?S }a=b+1$ where $c'$ is the predecessor of $c$, i.e.,
$c'+1=c$. For example, predicate $\text{?S }a=2$ is equivalent to
$\text{?S }\text{E}b\text{ }a=b+1 \mathbin{\&} b=1.$ Similarly
$\text{?S }b=1$ is equivalent to $\text{?S }\text{E}b_2\text{ } b=b_2+1
\mathbin{\&} b_2=0$. Based on Definition \ref{def:number systems}, for
all number systems $\mathbf{S}$, the automaton for $\text{?S }b_2=0$ is
the simple automaton accepting $0^*$. To construct automaton $(a,b):
\text{?S }a=b+1$, just note that the predicate is equivalent to
$\text{?S }\text{E}c\text{ } a=b+c \mathbin{\&} c=1$.

A similar recursive argument can be applied to obtain $(a,b): a=c*b$
for a constant $c>0$, i.e.,  one can construct $(a,b): (a=b_2+b)
\mathbin{\&} (b_2=c'*b)$ where $c'$ is the predecessor of $c$. The
similar argument can be applied to obtain automata for division by
constants or subtraction.

To construct $(a,b): \text{?S }a<=b$, note its equivalence to $(a,b):
\text{?S }a<b \mid a=b$. With similar arguments, one can construct
automaton for other comparison operators.

It is important to understand Walnut's construction of
$$(y,x_1,x_2,\ldots,x_n):\text{?S } y \olessthan (x_1 \otimes_1 x_2
\otimes_2 \cdots \otimes_{n-1} x_n)$$ where $n\geq 3$. Here
$\olessthan$ denotes an arbitrary comparison operator, and the
$\otimes_i$ are arbitrary arithmetic operators. Also let $y$ and
$x_i$ be variables or arithmetic constants. All arithmetic operators
in Walnut are associative from left to right; see Section
\ref{sec:operators}. Based on this, Walnut first transforms the
predicate to an equivalent predicate $$(y,x_1,x_2,\ldots,x_n):\text{?S
E}y_1,\ldots,y_{n-2}\text{ } (y_1=x_1 \otimes_1
x_2)\mathbin{\&}(y_2=y_1\otimes_2 x_3) \mathbin{\&}\cdots\mathbin{\&}
(y=y_{n-2}\otimes_{n-1} x_n).$$ Now Walnut has all the resources
necessary to construct this last automaton.

For example, to construct $(a): 0\text{<=}(a-1+1)$, Walnut first transforms it to $(a): \text{E}b\text{ }(b=a-1)\mathbin{\&}(0=b+1)$. The automaton is depicted below:
\begin{figure}[H]
	\centering
	\includegraphics[scale=.5]{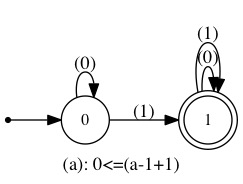}
	\caption{Automaton for $0\text{<=}a-1+1$ does not accept $0^*$}
	\label{fig:negative}
\end{figure}
There is something here that is worth noting. Note how this automaton
does not accept $0$? In arithmetic over integers $a=0$ satisfies the
predicate. However in Presburger arithmetic setting $a=0$ gives $b=-1$,
which is not acceptable, since Presburger arithmetic is defined over
natural numbers. In order to fix this issue, try to always postpone
subtraction and division to the rightmost position in your predicates.
For example, writing $(a): 0<=(a+1-1)$ results in
\begin{figure}[H]
	\centering
	\includegraphics[scale=.5]{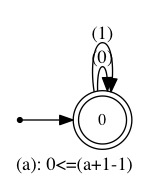}
	\caption{Automata for $0\text{<=}a+1-1$ accepts $0^*$}
	\label{fig:negative_fixed}
\end{figure}

\subsection{Calling an Automaton}\label{sec:calling an automaton}
In Section \ref{sec:calling expressions}, we learned about the syntax
and semantic of calling an automaton. A calling expression is a kind of
syntactic
sugar to save some space when writing long and complicated predicates.
Suppose we already have computed the automaton $(x_1,x_2,\ldots,x_n): P$ and
given it the name $M$. We can refer to $P$ in a predicate $P'$
without writing $P$ all over again, by just writing
$\$M(e_1,e_2,\ldots,e_n)$, where $\$$ symbol is to signify that $M$ is
an automaton, and the $e_i$ are either arithmetic expressions or
predicates with exactly one free variable. In such case, we say,
predicate $P'$ is calling $M$ (or is calling predicate $P$).

To construct automaton for $\$M(e_1,e_2,\ldots,e_n)$, Walnut constructs
the equivalent automaton:
$$\text{E}x_1,x_2,\ldots,x_n\text{ }P \mathbin{\&} (x_1=e'_1)\mathbin{\&} (x_2=e'_2)\mathbin{\&} \cdots\mathbin{\&}(x_n=e'_n)\mathbin{\&} (e_{j_1})\mathbin{\&} (e_{j_2})\mathbin{\&} \ldots\mathbin{\&} (e_{j_k})$$
where $x_1,x_2,\ldots,x_n$ are the free variables in $P$, $k$ is the
number of predicates in $e_1,e_2,\ldots,e_n$, $j_1,j_2,\ldots,j_k$ are
indices of predicates among $e_1,e_2,\ldots,e_n$, and if $e_j$ is an
arithmetic expression, then $e'_j=e_j$, otherwise $e_j$ is a predicate,
and $e'_j$ is the free variable occurring in $e_j$.

The fact that this predicate is equivalent to $\$M(e_1,e_2,\ldots,e_n)$
could be obtained easily using the semantic rule explained in Section
\ref{sec:calling expressions}. Walnut's implementation includes some
considerations to improve efficiency. For example, obviously when $e_j$
is a variable, we do not need to introduce a new variable $x_j$.

\begin{mdframed}[style=MyFrame]
Calling an automaton inside a predicate $P'$ is also more efficient than copying $P$ over and over again in $P'$. This is because Walnut does not need to construct $M$ every time we write $\$M$ in $P'$.
\end{mdframed}

The commands def and eval in Walnut are responsible for constructing the
automaton $M$ from predicate $P$. Unlike eval, the command def saves
the automaton $M$ so it can be called later from other predicates like
$P'$. See Section \ref{sec:def} for more information on def command.

To see an example, let $M$ be the automaton $(a,b):a+b=10$, and let $Q$
be the predicate $\$M(x,y)\mathbin{\&}y=8$. The automaton in Figure
\ref{fig:call_example} accepts $Q$.
\begin{figure}[H]
	\centering
	\includegraphics[scale=.5]{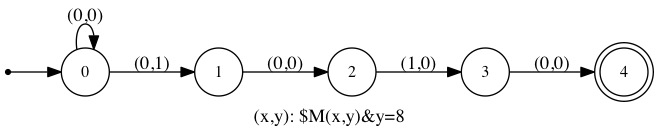}
	\caption{Automaton accepting $Q$}
	\label{fig:call_example}
\end{figure}

Please refer to Section \ref{sec:def}, which is devoted to examples of calling automata.  

\begin{mdframed}[style=MyFrame]
When calling an automaton $M$, one has to make sure that the $j$'th argument is in the same number system as $j$'th input in $M$ for all $j$.
\end{mdframed}

\subsection{Indexing an Automatic Word}\label{sec:indexing to an automatic word}

Suppose $W$ is an $n$-dimensional automatic word and
$M\big(Q,q_0,O,\delta,\Sigma,\mathbf{S_1},\mathbf{S_2},\ldots,\mathbf{S_n}\big)$
is its corresponding automaton with output. Also let $\alpha$ be an
alphabetic constant. We note that $(x_1,x_2,\ldots,x_n):
W[x_1][x_2]\cdots[x_n]=@\alpha$ is the automaton
$$\big(Q,q_0,F,\delta,\mathbf{S_1},\mathbf{S_2},\ldots,\mathbf{S_n}\big)$$
when minimized, where $F=\big\{q:O(q)=\alpha\big\}$. Similar arguments
can be made for other comparison operators.

Suppose $W_1$ and $W_2$ are $m$- and $n$-dimensional automatic words, respectively, and let $M_1$ and $M_2$ be their corresponding automata with output. We note that $(x_1,x_2,\ldots,x_m,y_1,y_2,\ldots,y_n): W_1[x_1][x_2]\cdots[x_m]=W_2[y_1][y_2]\cdots[y_n]$
is  $(M_1\times M_2)(F)$ when minimized, where $F$ contains all $(q_1,q_2)$ where $q_1$ and $q_2$ are states of $M_1$ and $M_2$, respectively, and they have the same output. Similar arguments can be made for other comparison operators.

The above statements can be proved easily using the semantic rule
explained in Section \ref{sec:indexing expressions}. Now what if
indices are arithmetic expressions and/or predicates with one free
variable? The construction is based on substitutions similar to the
ones mentioned for calling expressions in Section \ref{sec:calling an
automaton}.

\section{Special Automata in Walnut}\label{sec:special automata}

\subsection{True and False Automata}\label{sec:true and false}
In Section \ref{sec:automata accepting predicates} we saw an example of a predicate with no free variables: 
$$\text{A}x\text{ }\text{E}y\text{ } x=2*y \mid x=2*y+1$$
This predicate evaluates to true (it is a tautology). Here is an example of a predicate with no free variable that evaluates to false (contradiction):
$$\text{E}x\text{ } x>x+1$$
Walnut assigns a special automaton called true (false) automaton to predicates with no free variable that evaluate to true (false). However there could be predicates with free variables that are converted to true or false automata. See the following conventions implemented in Walnut:
\begin{itemize}
\item Conjunction (disjunction) of true automaton with automaton $M$ yields $M$ (true automaton, respectively).
\item Conjunction (disjunction) of false automaton with automaton $M$ yields false automaton ($M$, respectively).
\item Negation of true automaton is false automaton and vice versa.
\item Conventions for other logical operators follow from the above.
\end{itemize}
These conventions are reflecting the following facts from mathematical logic (for a predicate $P$):
\begin{itemize}
\item $P\mathbin{\&}\text{true}$ and $P\mid \text{true}$ are equivalent to $P$ and true respectively.
\item $P\mathbin{\&}\text{false}$ and $P\mid \text{false}$ are equivalent to false and $P$ respectively.
\item $\sim\text{true}=\text{false}$ and $\sim\text{false}=\text{true}$.
\end{itemize}
As an example, the automaton $(y): (\text{A}x\text{ } x<x+1)
\mathbin{\&} y=2$ is exactly the same as automaton $(y): y=2$. As
another example, the automaton $(y):(\text{E}x\text{ } x<0)
\mathbin{\&} y=2$ is the false automaton. As in our last example, note
that $\sim(\text{E}x\text{ } x<0)$ is the true automaton.

Figures~\ref{fig:true} and \ref{fig:false} show the special way Walnut
represents true and false automata.
\begin{figure}[h]
\centering
\begin{minipage}[t]{.5\textwidth}
\centering
\includegraphics[scale=.5]{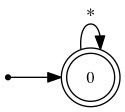}
\caption{True automaton}
\label{fig:true}
\end{minipage}%
\begin{minipage}[t]{.5\textwidth}
\centering
\includegraphics[scale=.5]{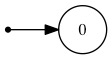}
\caption{False automaton}
\label{fig:false}
\end{minipage}
\end{figure}

\subsection{Non-arithmetic Automata}\label{sec:non-arithmetic automata}
There is a need for automata in which inputs (or some of them) do not
represent numbers in a specific number system. These automata might
accept patterns, or they might be relying on some non-arithmetic
instructions.

For example, the pattern $0^*10^*$ represents powers of $2$ in
$\MSD[2]$. However, the same pattern represents powers of $2$ in
$\LSD[2]$. Therefore, by not assigning a number system to the automaton
$M$ accepting the pattern $0^*10^*$, we are allowed to call $M$ both in
predicates in $\MSD[2]$ and in predicates in $\LSD[2]$. Assigning
number systems to automata accepting patterns usually does not make
much sense. See more examples in Section \ref{sec:reg}.

Allowing non-arithmetic automata is specially helpful when working with
the class of paperfolding words. These words are defined with an
automaton that takes two inputs. One input is a number that represents
a position in the paperfolding word and the other input is folding
instruction that does not represent numbers. To see how Walnut can be
used to prove properties of paperfolding words see
\cite{paperfolding}.

\section{Installation}\label{sec:installation}
Walnut is a command line program. You can run Walnut on any platform as long as you have Java 8 or later (preferably JDK 8 or higher) installed. To see which version of Java is installed on your machine type the following in the terminal (without the command line prompt \$):
\begin{lstlisting}[language=bash]
$java -version
\end{lstlisting}
If you download Walnut as the zipped file Walnut.zip first thing you need to do is to uncompress it. Then open the terminal (or command prompt in Windows), and change the directory to ``.../Walnut/bin/'', and run Walnut by typing:
\begin{lstlisting}[language=bash]
$java Main.prover
\end{lstlisting}
To exit Walnut, type the following command (with the semicolon):
\begin{lstlisting}[style=cmd]
exit;	
\end{lstlisting}
To make the distinction that we are typing a Walnut command, names of
all Walnut's commands are written in {\color{cmd} green}.  Walnut
produces graphical representations of automata among other things.
Those representations are files with .gv extensions. In order to open
these files you need to install Graphviz, a graph visualization package
which is available for all platforms.
\emph{All text files that Walnut produces are in the UTF-8 encoding.
All text files that Walnut reads have to be in the  UTF-8 encoding as well. Most editors, e.g., vim, notepad, etc., use UTF-8 encoding by default.}

\subsection{Eclipse}\label{sec:eclipse}
As explained in the previous section, you can use the terminal to work
with Walnut and enter your commands. However, I encourage you to use a
Java IDE, like Eclipse, because in my opinion, entering commands in the
console of a good IDE is more convenient than doing the same in the
terminal. You are only going to run Walnut inside the IDE and use the
IDE's console (not the source code editor) to enter Walnut commands.
Here is the instructions on how to run Walnut inside Eclipse for Java:
\begin{enumerate}
\item Go to this \href{https://www.eclipse.org/downloads/packages/eclipse-ide-java-developers/mars1}{link} and download Eclipse for Java for your specific platform.
\item Downloaded file is probably compressed. To start Eclipse, first uncompress the file, then click on the file named Eclipse. 
\item When you run Eclipse, it asks for a workspace address. Feel free to enter the path to your desired directory.

Now we need to import Walnut into Eclipse as a Java project:
\item When in Eclipse, go to ``File > Import ... ''. In the dialog that opens up choose ``General > Existing Project into Workspace''.
\item In the dialog that opens up, click browse, find Walnut (that you downloaded and uncompressed in the Installation section) and press open. Then click finish.
\item Close the Welcome page in your Eclipse window.
\item On the "Project Explorer" (probably) on the left of your screen, you can see only one project (the one that you just imported). Click on it. Then click on ``src''. Then double click on ``prover.java''. 
\item You will see a green circle with a white triangle inside it. Click on it. This causes Walnut to run. 
\item You can enter your Walnut commands in the console window in your Eclipse. If you are not able to find the console window, go to ``Window > Show View > Other > General > Console'' to open it. 
\end{enumerate}

\section{Commands}\label{sec:commands}
Every command ends in either a colon, double colon, or a semicolon. If you want to see the reports on the intermediate steps of a computation use colon or double colon, otherwise use semicolon. For example, if we type:
\begin{lstlisting}[style=cmd]
eval test "(*\textcolor{blue}{$a=b+1$}*)":
\end{lstlisting}
we get an output similar to the following written in the console:
\begin{lstlisting}[style=out]
(*\textcolor{out}{$a=b+1$}*):2 states - 1ms
total computation time: 2ms
\end{lstlisting}
which explains that the automaton for predicate $a=b+1$ has 2 states
and it took 2 milliseconds to compute it. We use {\color{blue}blue} to
denote predicates. Here we use {\color{out} grey} to indicate the
output produced by Walnut in the console. We use {\color{red} red} to
indicate errors in the console.

Walnut prints to console only the major steps of a computation when a single colon is used to end a command. Use two colons if you want to see all the steps behind the scene (cross product \ref{sec:cross product}, quantification \ref{sec:quantification}, minimizations, converting NFAs to DFAs (determinizations), etc.):

\begin{lstlisting}[style=cmd]
eval test "(*\textcolor{blue}{$a=b+1$}*)"::
\end{lstlisting}
then we get an output similar to the following written in the console:
\begin{lstlisting}[style=out]
computing (*\textcolor{out}{$b+1$}*)
computed (*\textcolor{out}{$a+1$}*)
computing (*\textcolor{out}{$a=(b+1)$}*)
	computing (*\textcolor{out}{$\&$}*):1 states - 2 states
	computing cross product:1 states - 2 states
	computed cross product:2 states - 1ms
		minimizing:2 states
			determinizing:2 states
			determinized:2 states - 0ms
		minimized:2 states - 0ms
	computed (*\textcolor{out}{$\&$}*):2 states - 1ms
	quantifying:2 states
		minimizing:2 states
			determinizing:2 states
			determinized:2 states - 0ms
		minimized:2 states - 0ms
	quantified:2 states - 0ms
	fixing leading zeros:2 states
		determinizing:2 states
		determinized:2 states - 1ms
		minimizing:2 states
			determinizing:2 states
			determinized:2 states - 0ms
		minimized:2 states - 0ms
	fixed leading zeros:2 states - 1ms
computed (*\textcolor{out}{$a=(b+1)$}*)
(*\textcolor{out}{$a=(b+1)$}*):2 states - 2ms
total computation time: 34ms
\end{lstlisting}

Whitespace is ignored. You can, for example, span one single command into multiple lines to improve readability. So, for example, you can write the following interchangeably:
\begin{lstlisting}[style=cmd]
eval test "(*\textcolor{blue}{$a=b+1$}*)";
eval test
"(*\textcolor{blue}{$a=b+1$}*)";
eval test "(*\textcolor{blue}{$a$}*)
(*\textcolor{blue}{$=b+1$}*)";
eval test
"(*\textcolor{blue}{$a=b+1$}*)"
;
\end{lstlisting}

In case we forget to separate the name test and predicate $a=b+1$ of the eval command, Walnut catches it by returning an error:
\begin{lstlisting}[style=cmd]
eval test"(*\textcolor{blue}{$a=b+1$}*)":
\end{lstlisting}
\vspace*{-\baselineskip}
\begin{lstlisting}[style=err]
invalid use of eval/def command
	: eval test"(*$a=b+1$*)":
\end{lstlisting}

Here is the full list of commands in Walnut and we will go over them one by one in detail:
\begin{itemize}
\item exit
\item eval <name> <predicate> 
\item def <name> <predicate>
\item macro <name> <template>
\item reg <name> <number system> <regular expression>
\item reg <name> <alphabet> <regular expression>
\item load <file name>
\end{itemize}

\subsection{eval: eval <name> <predicate>}\label{sec:eval}
This is the most important command in Walnut and it stands for ``evaluate.'' This command takes two arguments. The first argument is a name for the evaluation. Name of the evaluation starts with a letter and could contain alphanumerics and underscore. The files generated as the result of the eval command, all share the name given in the first argument. The second argument is a predicate that we want to evaluate. Predicates are always placed between quotation marks. To see the definition for predicates see Section \ref{sec:predicates}. In this article we typeset predicates in math mode in \LaTeX. However, the reader should note that this typesetting is different from the one they see in the terminal. Let us see an example:
\begin{lstlisting}[style=cmd]
eval four "(*\textcolor{blue}{$a=4$}*)":
(@(*\textcolor{out}{$a=4$}*): 4 states - 3ms
total computation time: 3ms@)
\end{lstlisting}
This evaluates to an automaton with one binary input labeled $a$. This
is the automaton $(a):a=4$. To learn about the notation $(a):a=4$ see
Section \ref{sec:automata accepting predicates}. The automaton accepts
only if $a$ is the most-significant-digit-first binary representation
of $4$, i.e., if it belongs to $0^*100$. This automaton is drawn and
saved in the directory ``/Walnut/Result/'' in a file named four.gv as
shown in Figure \ref{fig:four}. The graph drawing software Graphviz is
required to open this file; see Section \ref{sec:installation}.

\begin{figure}[H]
	\centering
	\includegraphics[scale=.5]{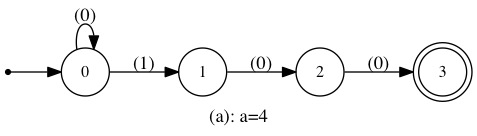}
	\caption{Content of the file four.gv}
	\label{fig:four}
\end{figure}

How does Walnut know to use the most-significant-digit-first binary
system? Walnut defaults to $\MSD[2]$ which is how we show the
most-significant-digit-first binary system in Walnut; see Section
\ref{sec:number systems} to learn about this notation and to learn
about number systems in general. To explicitly mention $\MSD[2]$ type:
\begin{lstlisting}[style=cmd]
eval four "(*\textcolor{blue}{$\text{?msd\_2 } a=4$}*)";
\end{lstlisting}
Similarly,  for the least-significant-digit-first binary type:
\begin{lstlisting}[style=cmd]
eval lsd_four "(*\textcolor{blue}{$\text{?lsd\_2 } a=4$}*)";
\end{lstlisting}
\begin{figure}[H]
	\centering
	\includegraphics[scale=.5]{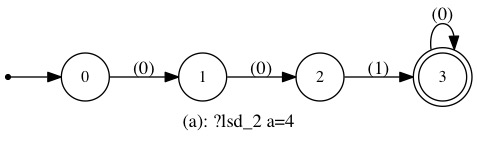}
	\caption{lsd\_four.gv}
	\label{fig:four in lsd 2}
\end{figure}

Here is another example, this time in $\LSD[3]$:
\begin{lstlisting}[style=cmd]
eval ternary_example "(*\textcolor{blue}{$\text{?lsd\_3 } a<5$}*)";
\end{lstlisting}
\begin{figure}[H]
	\centering
	\includegraphics[scale=.5]{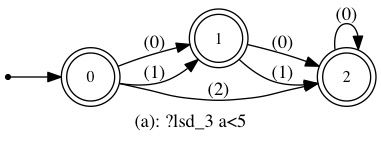}
	\caption{ternary\_example.gv}
	\label{fig:ternary example}
\end{figure}

This automaton accepts exactly those words representing the numbers
$0$,$1$,$2$,$3$,and $4$ in the least-significant-digit ternary base,
i.e., $0^*$,$10^*$,$20^*$,$010^*$,$110^*$ respectively. Note the
trailing zeros in the representations as opposed to the leading zeros
in a most-significant-digit-first ($\MSD$) number system. Also note
that this automaton accepts the empty word. This is because in the
definition of number systems we agreed that the empty word represents
$0$.

Let us see an example of an automaton with $2$ inputs:
\begin{lstlisting}[style=cmd]
eval two_inputs "(*\textcolor{blue}{$b=a+1$}*)";
\end{lstlisting}
This constructs the automaton $(a,b):b=a+1$ in which the first input corresponds to $a$, and the second input corresponds to $b$. Recall from Section \ref{sec:automata accepting predicates} that Walnut uses lexicographic ordering on the name of variables when constructing automata. So, for example, even though the first variable that appears in $b=a+1$ is $b$, it corresponds to the second input in the automaton.
\begin{figure}[H]
	\centering
	\includegraphics[scale=.5]{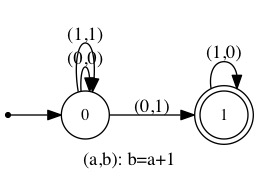}
	\caption{two\_inputs.gv}
	\label{fig:two inputs}
\end{figure}

Walnut generates two other files as the outcomes of the eval command
which can also be found in the directory \linebreak
``/Walnut/Result/''. For the evaluation two\_inputs, these two files
are named two\_inputs\_log.txt and two\_inputs.txt and they are both
text files.

The file two\_inputs\_log.txt contains the details of the evaluation
including the intermediate steps and the time each of those steps took
to complete. In our example, there are not many intermediate steps
involved:

\begin{figure}[H]
\begin{CenteredBox}
\begin{lstlisting}[style=cmd]
(*$b=2$*): 3 states - 0ms
total computation time: 0ms
\end{lstlisting}
\end{CenteredBox}
\captionof{lstlisting}{two\_inputs\_log.txt}
\end{figure}

In addition to two\_inputs\_log.txt , another file with the name two\_inputs\_detailed\_log.txt will be created if the command ends in two colons instead of a colon or a semicolon.

The file two\_inputs.txt contains the definition of the automaton in Figure \ref{fig:two inputs}:
\begin{figure}[H]
\begin{CenteredBox}
\begin{lstlisting}[style=file]
msd_2 msd_2
0 0
0 0 -> 0
0 1 -> 1
1 1 -> 0
1 1
1 0 -> 1
\end{lstlisting}
\end{CenteredBox}
\captionof{lstlisting}{two\_inputs.txt}
\label{file:two inputs}
\end{figure}

Line $1$ indicates that the first and the second inputs of the
automaton are both in $\MSD[2]$. The two states $0$ and $1$ in Figure
\ref{fig:two inputs} are declared in Lines $2$ and $6$. The first zero
in Line $2$ refers to the state $0$ and the second zero refers to its
output. Likewise, the first one in Line $6$ refers to the state $1$ and
the second one refers to its output. Note that the automaton for
evaluation two\_inputs is not an automaton with output, however all
automata are stored as automata with outputs in Walnut; see Section
\ref{sec:words and automata}. For an ordinary automaton, states with
non-zero outputs are interpreted as final states, and states with zero
outputs are interpreted as non-final states. So here state $0$ is
non-final, whereas state $1$ is final. Transitions for states $0$ and
$1$ are declared in Lines $3$-$5$ and $7$ respectively. For example,
state $0$ on $(0,0)$ transitions to itself, and on $(0,1)$ transitions
to state $1$. Transitions not depicted are transitions to the dead
state. For example, state $1$ transitions to the dead state on every
tuple except $(1,0)$. To learn more about definition of an automaton in
text files and how to manually define automata in text files see
Section \ref{sec:define automata in text files}.

In Section \ref{sec:automatic words}, we talked about the Thue-Morse
word. The Thue-Morse word's corresponding automaton with output,
depicted in Figure \ref{fig:thue}, is defined in directory
``/Walnut/Word Automata Library/'' in a file named T.txt. We can refer
to the Thue-Morse word in predicates by typing $T$. See Section
\ref{sec:define new automatic words} on how to define new automatic
words in Walnut.

We talked about square subwords in the Thue-Morse word. The following predicate is satisfied by $(i,n)$ if \linebreak$T[i..i+n-1]=T[i+n..i+2n-1]$, i.e., if there exists a square subword of length $2n$ starting at position $i$. 
\begin{lstlisting}[style=cmd]
eval squares_in_thue_morse_word "(*\textcolor{blue}{$n>0 \mathbin{\&}(\text{A}k\text{ } k< n\mathbin{\text{=>}}T[i+k]=T[i+n+k])$}*)";
\end{lstlisting}
The order of a square is half its length. Now if we want to find all natural numbers $n$ for which there exists a square of order $n$ in the Thue-Morse word, we simply use the existential quantifier E:
\begin{lstlisting}[style=cmd]
eval order_of_squares_in_thue_morse_word "(*\textcolor{blue}{$\text{E}i\text{ }n>0 \mathbin{\&}(\text{A}k\text{ } k< n\mathbin{\text{=>}}T[i+k]=T[i+n+k])$}*)";
\end{lstlisting}
\begin{figure}[H]
	\centering
	\includegraphics[scale=.5]{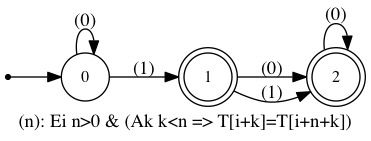}
	\caption{order\_of\_squares\_in\_thue\_morse\_word.gv}
	\label{fig:order of squares in thue morse word}
\end{figure}
Based on this automaton, the natural number $n$ with $\MSD[2]$ representation of the form $0^*(1|11)0^*$ is an order of a square in the Thue-Morse word. In other words, the set of orders in the Thue-Morse word is $$\big\{n:\text{there exists } k \geq 0\text{ such that }n=2^k \text{ or } n=2^{k+1}+2^{k}\big\}$$

Overlaps are the words of the form $axaxa$ where $a$ is a symbol and
$x$ is any word, e.g., the word ``alfalfa'' in English is an overlap. It is
a known that the Thue-Morse word avoids overlaps. How do we make sure,
using Walnut, that the Thue-Morse word does not have any overlaps? The
result of the following predicate must be the true automaton; see
Section \ref{sec:true and false}, if the Thue-Morse does not have any
overlaps:
\begin{lstlisting}[style=cmd]
eval thue_morse_does_not_have_overlaps "(*\textcolor{blue}{$\sim(\text{E}i,n\text{ }n>0 \mathbin{\&}(\text{A}k\text{ } k\mathbin{\text{<=}} n\mathbin{\text{=>}}T[i+k]=T[i+n+k]))$}*)":
(@(*\textcolor{out}{$n>0$}*): 2 states - 1ms
  (*\textcolor{out}{$k<=n$}*): 2 states - 1ms
    (*\textcolor{out}{$T[(i+k)]=T[((i+n)+k)]$}*): 12 states - 6ms
      (*\textcolor{out}{$(k\mathbin{\text{<=}}n\mathbin{\text{=>}}T[(i+k)]=T[((i+n)+k)])$}*): 25 states - 1ms
        (*\textcolor{out}{$(A k (k\mathbin{\text{<=}}n\mathbin{\text{=>}}T[(i+k)]=T[((i+n)+k)]))$}*): 1 states - 27ms
          (*\textcolor{out}{$(n>0\mathbin{\&}(A k (k\mathbin{\text{<=}}n\mathbin{\text{=>}}T[(i+k)]=T[((i+n)+k)])))$}*): 1 states - 0ms
            (*\textcolor{out}{$(E i , n (n>0\mathbin{\&}(A k (k\mathbin{\text{<=}}n\mathbin{\text{=>}}T[(i+k)]=T[((i+n)+k)]))))$}*): 1 states - 1ms
              (*\textcolor{out}{$\sim(E i , n (n>0\mathbin{\&}(A k (k\mathbin{\text{<=}}n\mathbin{\text{=>}}T[(i+k)]=T[((i+n)+k)]))))$}*): 1 states - 0ms
total computation time: 38ms@)
\end{lstlisting}
\begin{figure}[H]
	\centering
	\includegraphics[scale=.5]{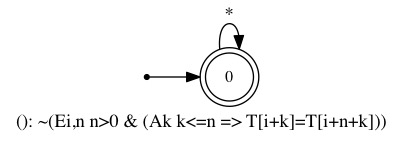}
	\caption{thue\_morse\_does\_not\_have\_overlaps.gv}
	\label{fig:thue_morse_does_not_have_overlaps}
\end{figure}
The automaton in Figure \ref{fig:thue_morse_does_not_have_overlaps} is the true automaton. For more information see Section \ref{sec:true and false}.

Note that if a predicate is not valid Walnut returns an error:
\begin{lstlisting}[style=cmd]
eval invalid "(*\textcolor{blue}{$x+y+z$}*)";
\end{lstlisting}
\vspace*{-\baselineskip}
\begin{lstlisting}[style=err]
the final result of the evaluation is not of type automaton
	: eval invalid "(*$x+y+z$*)";
\end{lstlisting}

To understand why this is not a valid predicate see Section \ref{sec:arithmetic expressions}. In the following examples note how Walnut points to the locations of the errors in the predicates. By saying "char at $n$", Walnut tries to convey that there is something wrong going on at the vicinity of the $n$'th character in the predicate.
\begin{lstlisting}[style=cmd]
eval invalid2 "(*\textcolor{blue}{$(x+y+z=0$}*)";
\end{lstlisting}
\vspace*{-\baselineskip}
\begin{lstlisting}[style=err]
unbalanced parenthesis
	: char at 0
	: eval invalid2 "(*$(x+y+z=0$*)";
\end{lstlisting}
\begin{lstlisting}[style=cmd]
eval invalid3 "(*\textcolor{blue}{$({\sim x})=0$}*)";
\end{lstlisting}
\vspace*{-\baselineskip}
\begin{lstlisting}[style=err]
operator (*$\sim$*) cannot be applied to the operand (*$x$*) of type variable
: char at 1
: eval invalid3 "(*$({\sim x})=0$*)";
\end{lstlisting}
\begin{lstlisting}[style=cmd]
eval invalid4 "(*\textcolor{blue}{$T[i+j]=i-1$}*)";
\end{lstlisting}
\vspace*{-\baselineskip}
\begin{lstlisting}[style=err]
operator = cannot be applied to operands (*$T[(i+j)]$*) and (*$(i-1)$*) of types word and arithmetic respectively
	: char at 6
	: eval invalid4 "(*$T[i+j]=i-1$*)";
\end{lstlisting}
\begin{lstlisting}[style=cmd]
eval invalid5 "(*\textcolor{blue}{$T[2]=1$}*)";
\end{lstlisting}
\vspace*{-\baselineskip}
\begin{lstlisting}[style=err]
operator = cannot be applied to operands (*$T[2]$*) and (*$1$*) of types word and number literal respectively
	: char at 4
	: eval invalid5 "(*$T[2]=1$*)";
\end{lstlisting}
The last example can be fixed as follows:
\begin{lstlisting}[style=cmd]
eval fixed5 "(*\textcolor{blue}{$T[2]=@1$}*)";
\end{lstlisting}
To understand why see Section \ref{sec:relative expressions}.

The last thing to note about the eval command is that Walnut overrides the files generated by an evaluation if the name of the evaluation is used in a new evaluation.

\subsection{def: def <name> <predicate>}\label{sec:def}
The word def stands for define. The syntax for this command is exactly the same as the syntax for eval command. The only difference between this command and eval is that the automaton constructed is saved in the directory ``/Walnut/Automata Library/'' for later use.  Suppose we write the following:
\begin{lstlisting}[style=cmd]
def sum10 "(*\textcolor{blue}{$x+y=10$}*)";
\end{lstlisting}
This creates as usual the files sum10.gv, sum10.txt, and sum10\_log.txt in the directory ``/Walnut/Result/''. However, it also saves a copy of sum10.txt in the directory ``/Walnut/Automata Library/''. Any automaton saved in this directory can be called in other predicates by referring to its name and the special character \$. To learn about calling see Sections \ref{sec:calling expressions} and \ref{sec:calling an automaton}.

Let us see examples of predicates calling the automaton sum10:
\begin{lstlisting}[style=cmd]
eval lessThanThree "(*\textcolor{blue}{$\text{E}a\text{ }a\mathbin{\text{>=}}8 \mathbin{\&} \$\text{sum10}(b,a)$}*)";
\end{lstlisting}
This predicate is satisfied by numbers $b$ for which there exist an $a\mathbin{\text{>=}}8$ such that $b+a=8$, i.e., $0,1,2$:
\begin{figure}[H]
	\centering
	\includegraphics[scale=.5]{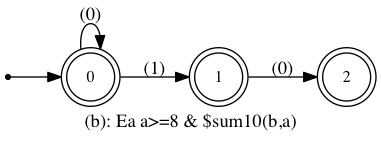}
	\caption{lessThanThree.gv}
	\label{fig:lessThanThree}
\end{figure}

We can send the same variable to both arguments of sum10:
\begin{lstlisting}[style=cmd]
eval five "(*\textcolor{blue}{$\$\text{sum10}(a,a)$}*)";
\end{lstlisting}
\begin{figure}[H]
	\centering
	\includegraphics[scale=.5]{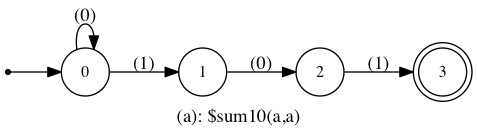}
	\caption{five.gv}
	\label{fig:five}
\end{figure}

We can send constants to any arguments of sum10:
\begin{lstlisting}[style=cmd]
eval three "(*\textcolor{blue}{$\$\text{sum10}(7,a)$}*)";
\end{lstlisting}
\begin{figure}[H]
	\centering
	\includegraphics[scale=.5]{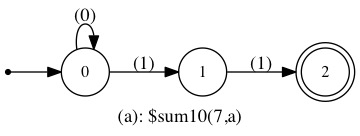}
	\caption{three.gv}
	\label{fig:three}
\end{figure}

Indeed, we can send any arithmetic expressions or predicates with one free variable to arguments:
\begin{lstlisting}[style=cmd]
eval three "(*\textcolor{blue}{$\$\text{sum10}(a-2,3*a)$}*)";
\end{lstlisting}
\begin{lstlisting}[style=cmd]
eval three "(*\textcolor{blue}{$\text{E}b\text{ }\$\text{sum10}(a,b+3=10)$}*)";
\end{lstlisting}
The resulting automaton for both of these is depicted in Figure \ref{fig:three}. We can call sum10 to define new automata:
\begin{lstlisting}[style=cmd]
def threeSum10 "(*\textcolor{blue}{$\$\text{sum10}(x+y,z)$}*)";
\end{lstlisting}
Now we can write 
\begin{lstlisting}[style=cmd]
eval three "(*\textcolor{blue}{$\text{E}y,z\text{ }\$\text{threeSum10}(x,y,z) \mathbin{\&} y=2 \mathbin{\&} z=5$}*)";
\end{lstlisting}
The result of this evaluation is again depicted in Figure \ref{fig:three}.

Now look at the following example:
\begin{lstlisting}[style=cmd]
eval nonsense "(*\textcolor{blue}{$\$\text{sum10}(a=b,4)$}*)";
\end{lstlisting}
\vspace*{-\baselineskip}
\begin{lstlisting}[style=err]
argument 1 of function sum10 cannot be an automaton with != 1 inputs
	: char at 1
	: eval nonsense "(*$\$\text{sum10}(a=b,4)$*)";
\end{lstlisting}
This is because the first argument is a predicate with two free variables.

We cannot send a variable in $\LSD[2]$ to an automaton that accepts only $\MSD[2]$, and expect getting anything interesting in return. The following example would run fine, but the result is another nonsense:
\begin{lstlisting}[style=cmd]
eval another_nonsense "(*\textcolor{blue}{$\text{?lsd\_2 }\$\text{sum10}(x,4)$}*)";
\end{lstlisting}
\begin{figure}[H]
	\centering
	\includegraphics[scale=.5]{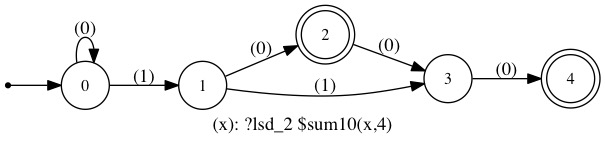}
	\caption{another\_nonsense.gv}
	\label{fig:another nonsense}
\end{figure}

The number of arguments when calling an automaton should match the number of inputs of that automaton:
\begin{lstlisting}[style=cmd]
eval invalid "(*\textcolor{blue}{$\$\text{sum10}(x,y,z)$}*)";
\end{lstlisting}
\vspace*{-\baselineskip}
\begin{lstlisting}[style=err]
function sum10 requires 2 arguments
	: char at 1
	: eval invalid "(*$\$\text{sum10}(x,y,z)$*)";
\end{lstlisting}

Always remember the roles of the inputs to an automaton created by the def command. For example, look at the following examples of the def command:
\begin{lstlisting}[style=cmd]
def f1 "(*\textcolor{blue}{$y < x$}*)";
def f2 "(*\textcolor{blue}{$x < y$}*)";
\end{lstlisting}
Now the following evaluates to an automaton accepting representations of numbers $>1$:
\begin{lstlisting}[style=cmd]
eval greater_than_1 "$f1(a,1)";
\end{lstlisting}
whereas the following evaluates to an automaton accepting representations of numbers $<1$: 
\begin{lstlisting}[style=cmd]
eval less_than_1 "$f2(a,1)";
\end{lstlisting}
This is because f1 is an automaton for which the first argument is greater than the second argument, whereas, f2 is an automaton for which the first argument is less than the second argument. \emph{Always remember that Walnut sorts inputs of an automaton based on their labels' lexicographic order.}

\subsection{macro: macro <name> <template>}
Recall that a subword $W[i..i+n-1]$ is a palindrome if $W[i+k]=W[i+n-1-k]$ for all $k<n$. The following calculates the palindromes in the Thue-Morse word\footnote{In specific terms, this calculation constructs an automaton accepting $(i,n)$'s for which $W[i..i+n-1]$ is a palindrome.}:
\begin{lstlisting}[style=cmd]
eval thue_pal "Ak k<n => T[i+k] = T[i+n-1-k]";
\end{lstlisting}
Now in order to calculate palindromes for the Fibonacci word, one has to type the entire predicate above, letter for letter, except for the substitution of $T$ by $F$:
\begin{lstlisting}[style=cmd]
eval fibonacci_pal "?msd_fib Ak k<n => F[i+k] = F[i+n-1-k]";
\end{lstlisting}
There are many other example where one would like to calculate the same calculation for many different words. Similar to the def command, the macro command is a mechanism for avoiding dull repetitions. For the palindrome example, one could define the following macro:
\begin{lstlisting}[style=cmd]
macro pal "?%0 Ak k<n => %1[i+k] = %1[i+n-1-k]";
\end{lstlisting}
This command saves the template \lstinline|?\%0 Ak k<n => \%1[i+k] = \%1[i+n-1-k]| in a file named pal.txt in the directory ``/Walnut/Macro Library/''.
Then to calculate palindromes for the Thue-Morse word, using a similar syntax to function calls, one could write: 
\begin{lstlisting}[style=cmd]
eval thue_pal "#pal(msd_2,T)";
\end{lstlisting}
The first step in computing eval and def commands is a preprocessing task in which all macros will be replaced by their template strings. In our example, Walnut replaces the \lstinline|#pal(msd_2,T)| by the template, saved in the file pal.txt, then substitutes \lstinline|%0| and \lstinline|%1| by \lstinline|msd_2| and \lstinline|T|, respectively. Now see how easily one can calculate plaindromes for the Fibonacci word:
\begin{lstlisting}[style=cmd]
eval fibonacci_pal "#pal(msd_fib,F)";
\end{lstlisting}

The important difference between the macro and the def commands is that the macro takes a template and stores it unaltered as a string, whereas the def command takes a predicate and stores the automaton calculated for that predicate. The macro saves the user from typing the same template over and over, while the def command saves the user from calculating the same predicate over and over.

Let us see a second example. An important automaton for any word is one that accepts triples $(i,j,n)$ for which subwords of length $n$ starting at positions $i$ and $j$ are the same. This automaton, which we call factoreq, shows up in many different calculations, a notable one being border calculation. Recall that $W[i..i+n-1]$ has a border of length $m$ if subwords $W[i..i+m-1]$ and $W[i+n-m..i+n-1]$ are equal. In other words, the $W[i..i+n-1]$ has a border of length $m$ if factoreq accepts $(i,i+n-m,m)$. In Walnut, for Thue-Morse, this is the same as writing:
\begin{lstlisting}[style=cmd]
def thue_factoreq "Ak k<n => T[i+k]=T[j+k]";
eval thue_border "$1 <= m & m <= n & $thue_factoreq(i,i+n-m,m)";
\end{lstlisting}
Now to calculate for Fibonacci, this becomes:
\begin{lstlisting}[style=cmd]
def fibonacci_factoreq "?msd_fib Ak k<n => F[i+k]=F[j+k]";
eval fibonacci_border "?msd_fib 1 <= m & m <= n & $fibonacci_factoreq(i,i+n-m,m)";
\end{lstlisting}
Using macros, one could capture the essence of the above calculations as follows:
\begin{lstlisting}[style=cmd]
macro factoreq "?%0 Ak k<n => %1[i+k]=%1[j+k]";
macro border "?%0 1 <= m & m <= n & #%1_factoreq(i,i+n-m,m)";
\end{lstlisting}
With these macros at hand, one can calculate factoreq and border for many different words. For example, for Thue-Morse the calculation becomes:
\begin{lstlisting}[style=cmd]
def thue_factoreq "#factoreq(msd_2,T)";
eval thue_border "#border(msd_2,thue)";
\end{lstlisting}
and for the Fibonacci it becomes:
\begin{lstlisting}[style=cmd]
def fibonacci_factoreq "#factoreq(msd_fib,F)";
eval fibonacci_border "#border(msd_fib,fibonacci)";
\end{lstlisting}

\begin{mdframed}[style=MyFrame]
	Walnut does not accept nested macro/function calls, i.e.,  none of the followings are accepted: \lstinline|$f(...,$g(...),...)|, \lstinline|$f(...,#m(...),...)|, \lstinline|#m(...,#n(...),...)|, \lstinline|#m(...,$g(...),...)|.
\end{mdframed}

\subsection{Matrices with eval and def}
We postponed discussing the exact syntax of the eval and def commands until this section. These two commands can take an optional argument:
\begin{itemize}
	\item eval <name> <space separated list of free variables>? <predicate> 
	\item def <name> <space separated list of free variables>? <predicate>
\end{itemize}

The question marks are representing the fact that the corresponding arguments are optional. When a list of free variables is preset, in addition to the ordinary computation, Walnut will calculate the \emph{incidence matrices} of the underlying graph of the computed automaton, corresponding to the transitions labeled by the listed free variables. It stores the matrices in the Maple syntax in a .mpl file in ``/Walnut/Result/'' directory.

Suppose $M\big(Q,q_0,F,\delta,\mathbf{S_1},\ldots,\mathbf{S_m}\big)$ is the automaton calculated for the predicate $(x_1,\ldots,x_n):P'$. For a free variable $x_i$ and a value $v \in \Sigma_{\mathbf{S_i}}$, we let $M\_x_i\_v$ denote the $|Q|\times |Q|$ matrix where the entry $M\_x_i\_v[p][q]$\footnote{$M\_x_i\_v[p][q]$ denotes the entry at $p$th row and $q$th column.}, for every pair of states $p$ and $q$,  is the number of tuples transitioning from $p$ to $q$ with $v$ in the $i$'th coordinate, or more formally
$$M\_x_i\_v[p][q] = \left|\left\{\left(\alpha_1,\ldots,\alpha_i,\ldots,\alpha_m\right): \delta\left((\alpha_1,\ldots,\alpha_i,\ldots,\alpha_m),p\right)=q \text{ such that } \alpha_i=v \text{ and } \alpha_j\in  \Sigma_{\mathbf{S_j}} \text{ for } j\neq i \right\}\right|.$$

We already saw in Figure \ref{fig:two inputs} an automaton accepting $b=a+1$. Now writing
\begin{lstlisting}[style=cmd]
eval example a "b=a+1";
\end{lstlisting}
causes Walnut to calculate the matrices $M\_a\_0$ and $M\_a\_1$:

\[
M\_a\_0 = 
\begin{bmatrix}
1 & 1 \\
0 & 0
\end{bmatrix}
\]
\[
M\_a\_1 =
\begin{bmatrix}
1 & 0 \\
0 & 1
\end{bmatrix}
\]

In the same command, we could ask Walnut to calculate the incidence matrices for $b$ as well as $a$ by writing:
\begin{lstlisting}[style=cmd]
eval example a "b=a+1";
\end{lstlisting}
\[
M\_b\_0 = 
\begin{bmatrix}
1 & 0 \\
0 & 1
\end{bmatrix}
\]
\[
M\_b\_1 =
\begin{bmatrix}
1 & 1 \\
0 & 0
\end{bmatrix}.
\]

There are a lot of interesting things one can do using incidence matrices.
For example for Fibonacci word, and Sturmian words in general, the number of distinct subwords of length $n$, i.e., subword complexity, is $n+1$. Now can one verify this using Walnut? The answer is yes. Please refer to \cite{fib3} for this and a lot of other interesting applications.

\subsection{reg}\label{sec:reg}
The word reg stands for regular expression. Before we talk about this command in detail, let us motivate the need for it through an example. Suppose we need an automaton accepting $\MSD[2]$ representations of powers of $2$ that are less than $20$. There is no straightforward way of constructing an automaton accepting representations of powers of $2$ using eval and def commands\footnote{The set of powers of $2$ is not expressible in Presburger arithmetic. However the extended Presburger arithmetic that involves automatic words is powerful enough to express this set (why?)}. Remember how def command saves automata definition in directory ``/Walnut/Automata Library/''? We can manually create a file power2.txt in this directory and write in it the definition of an automaton accepting binary representations of powers of $2$:
\begin{figure}[H]
\begin{CenteredBox}
\begin{lstlisting}[style=file]
msd_2
0 0
0 -> 0
1 -> 1
1 1
0 -> 1
\end{lstlisting}
\end{CenteredBox}
\captionof{lstlisting}{power2.txt}
\label{file:power2}
\end{figure}
See Section \ref{sec:define automata in text files} to learn the syntaxes of defining an automaton in a text file. Now we can write a predicate for powers of $2$ that are less than $20$:
\begin{lstlisting}[style=cmd]
eval power2LessThan20 "(*\textcolor{blue}{$\$\text{power2}(a) \mathbin{\&} a<20$}*)";
\end{lstlisting}
\begin{figure}[H]
	\centering
	\includegraphics[scale=.5]{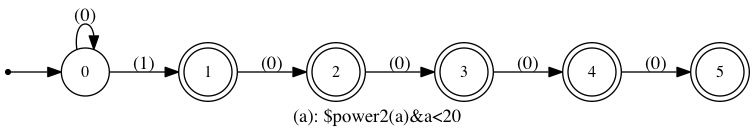}
	\caption{power2LessThan20.gv}
	\label{fig:power2LessThan20}
\end{figure}
The better approach to this problem is to use the reg command. This command can be used in two different ways:
\begin{enumerate}
\item reg <name> <number system> <regular expression>
\item reg <name> <alphabet> <regular expression>
\end{enumerate}
To construct an automaton that accepts $\MSD[2]$ representations of powers of $2$ we can use the first syntax:
\begin{lstlisting}[style=reg]
reg power2 msd_2 "(@\textcolor{blue}{$0^*10^*$}@)";
\end{lstlisting}
Similar to eval and def command, the second argument is a name. The
third argument is a number system, and the last argument is a regular
expression.  This will construct an automaton for the regular
expression, and saves the result in a file named power2.txt in
``/Walnut/Automata Library/'', in addition to saving, as usual, the
drawing of the automata in power2.gv in directory ``/Walnut/Result/''.
The file power2.txt is exactly the same as File \ref{file:power2}.

Note that $0^*10^*$ is also the $\LSD[2]$ representations of powers of
$2$. For this reason, there needs to be a way of defining an automaton
from a regular expression that is not restricted to a particular number
system. We call such an automaton a non-arithmetic automaton; see
Section \ref{sec:non-arithmetic automata}. To create a non-arithmetic
automaton accepting a pattern we can use the second version of the reg
command in which instead of a number system we specify an alphabet:
\begin{lstlisting}[style=reg]
reg general_power2 {0,1} "(@\textcolor{blue}{$0^*10^*$}@)";
\end{lstlisting}
The file general\_power2.txt generated by this command is the following:
\begin{figure}[H]
\begin{CenteredBox}
\begin{lstlisting}[style=file]
{0,1} 
0 0
0 -> 0
1 -> 1
1 1
0 -> 1
\end{lstlisting}
\end{CenteredBox}
\captionof{lstlisting}{general\_power2.txt}
\label{file:general_power2}
\end{figure}

The only difference between Files \ref{file:power2} and
\ref{file:general_power2} is the first line; see Section
\ref{sec:define automata in text files} for more information.

Since general\_power2 is not restricted to a particular number system both of the following are valid:
\begin{lstlisting}[style=cmd]
eval power2Less20_msd "(*\textcolor{blue}{$\text{?msd\_2 }\$\text{general\_power2}(a)\mathbin{\&} a<20$}*)";
eval power2Less20_lsd "(*\textcolor{blue}{$\text{?lsd\_2 }\$\text{general\_power2}(a)\mathbin{\&} a<20$}*)";  
\end{lstlisting}

Note that $0^*10^*$ is also $\MSD[n]$ and $\LSD[n]$ representations of powers of $n$ for any $n>1$. So what if we write the following:
\begin{lstlisting}[style=cmd]
eval invalid "(*\textcolor{blue}{$\text{?msd\_3 }\$\text{general\_power2}(a)\mathbin{\&} a<20$}*)";
\end{lstlisting}
\vspace*{-\baselineskip}
\begin{lstlisting}[style=regerr]
in computing cross product of two automata, variables with the same label must have the same alphabet
	: char at 12
	: eval power3_less10_msd "(@$\text{?msd\_3 }\$\text{general\_power2}(a)\mathbin{\&} a<20$@)";
\end{lstlisting}

Here Walnut is complaining about the fact that $\MSD[3]$'s alphabet is
$\{0,1,2\}$, whereas general\_power2's input alphabet is $\{0,1\}$; see
File \ref{file:general_power2}. Walnut is very strict about matching
alphabets, which we understand is sometimes a drawback, for example in
the above example. We will improve this feature in future releases of
Walnut.

We use the automata library in \cite{dk.brics} for converting regular
expressions to automata. To see the syntax for regular expressions
refer to this
\href{https://www.brics.dk/automaton/doc/index.html}{website}.

Here is a summary of the important syntax:
\begin{table}[h]
\centering
\begin{tabular}{| m{1.7cm} | m{9cm} |} 
\hline
$*$  & zero or more occurrences of an expression\\
\hline
$+$  & one or more occurrences of an expression\\
\hline
$\mid$   & union, e.g., $(0\mid1)2^*$\\
\hline
$.$   & any single character, e.g., $2.^*$\\
\hline
$[]$ & character class, e.g., $[1-4]$ means any of $1,2,3,4$\\
\hline
$\wedge$& complement of a character class, e.g., $[^\wedge2-9]$ is any of $0,1$\\
\hline
\end{tabular}
\caption{syntax summary for regular expressions}
\label{tab:regular expressions}
\end{table}

The alphabet in the second version of reg command could only be a subset of $\{0,1,\ldots,9\}$. Therefore the following is not allowed:
\begin{lstlisting}[style=reg]
reg invalid {0,-1,-2} "(@\textcolor{blue}{$-20^*$}@)";
\end{lstlisting}
\vspace*{-\baselineskip}
\begin{lstlisting}[style=regerr]
the input alphabet of an automaton generated from a regular expression must be a subset of {0,1,...,9}
	: reg invalid {0,-1,-2} "(@$-20^*$@)";
\end{lstlisting}

The last thing to note about the reg command is that for any regular
expression $r$, the resulting automaton from reg command is the
intersection of the automaton for $r$ with
$\Sigma^*$ where $\Sigma$ is the alphabet given as the third argument of the reg command. For example:

\begin{lstlisting}[style=reg]
reg note_the_intersection {2,3} "(@\textcolor{blue}{$2.^*2$}@)";
\end{lstlisting}
\begin{figure}[H]
	\centering
	\includegraphics[scale=.5]{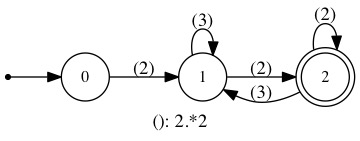}
	\caption{note\_the\_intersection.gv}
	\label{fig:note the intersection}
\end{figure}

\subsection{load: load <file name>}\label{sec:load}
We can write any series of legitimate Walnut commands in a text file and save it in the directory ``/Walnut/Command Files/''. Then we can load it by writing the following in Walnut:
\begin{lstlisting}[style=reg]
load file_name.txt;
\end{lstlisting}
This runs all commands in file\_name.txt in the order they appear. 
Recall that the file's encoding must be UTF-8. 

\section{Working with Input/Output}\label{sec:working with input and output}
Throughout this section it is assumed that all files have UTF-8 encoding. This is the default encoding for most text editors.

\subsection{Defining Automata in Text Files}\label{sec:define automata in text files}
In this section we learn how to manually define all automata types in
text files. Recall that an ordinary automaton can be thought of as an
automaton with output, in which states with non-zero outputs are
treated as final states; see Section \ref{sec:words and automata}.
Therefore suppose
$M\big(Q,q_0,O,\delta,\Sigma,\mathbf{S_1},\mathbf{S_2},\ldots,\mathbf{S_n}\big)$
is an automaton with output with $m$ states and $n$ inputs over number
systems $\mathbf{S_i}$. Furthermore suppose that the states are labeled
$0$ to $m-1$, i.e., $Q=\big\{0,1,\ldots,m-1\big\}$, and that
$q_0=0$\footnote{If an automaton does not follow these criteria we can
always come up with an isomorphic one that does.}. To define $M$ in a
text file, first create a text file M.txt\footnote{File names in Walnut
start with letters and can contain alphanumerics and underscore.}. The
first line must be $$S_1\text{  }S_2\text{  }\cdots\text{  }S_n$$which
declares inputs' number systems. The second line is declaring state $0$
as follows: $$0\text{  }\alpha$$ where $\alpha=O(0)$. Next lines are
declarations of transitions of state $0$ which can come in any order.
For every $\alpha_1\in \Sigma_{\mathbf{S_1}}, \alpha_2\in
\Sigma_{\mathbf{S_2}}, \ldots ,\alpha_n\in \Sigma_{\mathbf{S_n}}$
transitions are of the following form
$$\alpha_1\text{  }\alpha_2\text{  } \cdots \text{  }\alpha_n \mathbin{\text{->}} q$$
whenever $\delta\big(0,\alpha_1,\alpha_2,\ldots,\alpha_n\big) = q$.
There is no need to declare transitions to a dead state. For any pair
$(\alpha_1,\alpha_2,\ldots,\alpha_n)$ that no declaration of the form
above is mentioned, it is assumed that
$\delta(0,\alpha_1,\alpha_2,\ldots,\alpha_n)$ is a dead state. We can
use $*$, the wildcard matching symbol, in place of any symbol
$\alpha_i$. If there is a transition of the form $$\alpha_1\text{  }
\alpha_2\text{  } \ldots\text{  } \alpha_{i-1}\text{  } *\text{  }
\alpha_{i+1}\text{  }\ldots\text{  }\alpha_n \mathbin{\text{->}} q$$ it
is understood that
$\delta\big(0,(\alpha_1,\alpha_2,\ldots,\alpha_{i-1},\beta,\alpha_{i+1},\ldots,\alpha_n)\big)
= q$ for every $\beta \in \Sigma_{\mathbf{S_i}}$.  After transitions of
the state $0$ are declared, we declare state $1$ followed by its
transitions. We continue like this until all states and their
transitions are declared. Note that nowhere in M.txt we are defining
the output alphabet $\Sigma$. The output alphabet is inferred
indirectly by looking at the state declarations. To see examples refer
to Files \ref{file:two inputs}--\ref{file:general_power2}.

A non-arithmetic automaton is defined in the same way, except that in
the first line, for inputs that do not have number systems associated
with them, we write down the alphabet between curly brackets. Alphabets
can be any subset of integers. As an example see File
\ref{file:general_power2}.

Defining true or false automata in text files is easy. They have only
one line and it is either true or false.

As one last example, the paperfolding words are given by the following automaton; see Section \ref{sec:non-arithmetic automata} and article \cite{paperfolding} for more details:
\begin{figure}[H]
	\centering
	\includegraphics[scale=.5]{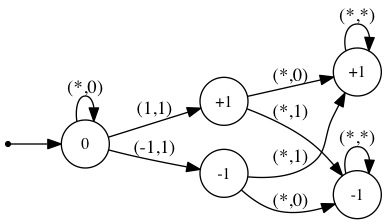}
	\caption{Automata for paperfolding words}
	\label{fig:paperfoldings}
\end{figure}
This automaton is defined in the file PF.txt in directory ``/Walnut/Word Automata Library/'':
\begin{figure}[H]
\begin{CenteredBox}
\begin{lstlisting}[style=file]
{-1,1} lsd_2

0 0
* 0 -> 0
1 1 -> 1
-1 1 -> 2

1 1
* 1 -> 4
* 0 -> 3

2 -1
* 0 -> 4
* 1 -> 3

3 1
* * -> 3

4 -1
* * -> 4
\end{lstlisting}
\end{CenteredBox}
\captionof{lstlisting}{PF.txt}
\label{file:paperfoldings}
\end{figure}

\subsection{Defining New Automatic Words}\label{sec:define new automatic words}
The eval and def commands always produce automata accepting a predicate, therefore the result is never an automaton with output. So to define an automatic word $W$, we need to manually define its corresponding automaton with output in the directory ``/Walnut/Word Automata Library/''. For example, the Thue-Morse word is defined in the file ``/Walnut/Word Automata Library/T.txt'' as follows:
\begin{figure}[H]
\begin{CenteredBox}
\begin{lstlisting}[style=file]
msd_2
0 0
0 -> 0
1 -> 1
1 1
0 -> 1
1 -> 0
\end{lstlisting}
\end{CenteredBox}
\captionof{lstlisting}{T.txt}
\label{file:thue morse}
\end{figure}

\subsection{Defining New Number Systems}\label{sec:define new number systems}

Based on Definition \ref{def:walnut_number_system}, to define a new
number system $\mathbf{S}$, we need to define automata for
$R_{\mathbf{S}}$, $+_{\mathbf{S}}$, and $<_{\mathbf{S}}$. We do not
need to define an automaton for $=_{\mathbf{S}}$, because it can be
generated easily, since we assumed that $w_1=_{\mathbf{S}} w_2$ if and
only if $w_1 = w_2$ for any two words $w_1$ and $w_2$ of the same
length. The automata for number systems must be defined in the
directory ``/Walnut/Custom Bases/''. For example, for the number system
$\mathbf{S}$, assuming it is $\MSD$, one needs to create msd\_S.txt,
msd\_S\_addition.txt, and msd\_S\_less\_than.txt for $R_{\mathbf{S}}$,
$+_{\mathbf{S}}$, and $<_{\mathbf{S}}$ respectively.  If $\mathbf{S}$
is $\LSD$, file names must be lsd\_S.txt, lsd\_S\_addition.txt, and
lsd\_S\_less\_than.txt respectively. The number system $\mathbf{S}$
defined in this way can be used in predicates by typing ?msd\_S or
?lsd\_S depending on whether $\mathbf{S}$ is $\MSD$ or $\LSD$. If the
automaton for $<_{\mathbf{S}}$ is not defined by the user, Walnut
assumes that $<_{\mathbf{S}}$ is the lexicographic ordering, i.e., if
$w_1$ and $w_2$ are of the same length, then $w_1 <_{\mathbf{S}} w_2$
if and only if $w_1$ comes before $w_2$ in lexicographic
order\footnote{Lexicographic ordering on symbols is assumed to be
$\cdots<-2<-1<0<1<2<\cdots$.}. If the automata for $R_{\mathbf{S}}$ is
not given, then $R_{\mathbf{S}}$ is assumed to be
$\Sigma_{\mathbf{S}}^*$. The alphabet $\Sigma_{\mathbf{S}}$ is inferred
from the automaton for $+_{\mathbf{S}}$ which is always given.

Note that reversing all automata for $\MSD[n]$ we get the corresponding
automata for $\LSD[n]$. The same goes with $\MSD[fib]$ and $\LSD[fib]$.
Thus for a number system $\MSD[S]$ if we only define files for
$\MSD[S]$, but then typing ?lsd\_S in a predicate, Walnut automatically
creates automata for $\LSD[S]$ by reversing those of $\MSD[S]$ and vice
versa. However the user should be cautious since there could very well
be number systems for which the difference between $\MSD$ and $\LSD$ is
more than the direction of the arrows in their corresponding automata.

\subsection{Converting .gv files to .jpeg}
The drawings of automata in Walnut are stored in .gv files. Not only can the software Graphviz open the files with this extension,
but it can also convert them to many different file formats. For example, suppose you have a file named automaton.gv. To convert it to automaton.jpeg type the following in the terminal:
\begin{lstlisting}[language=bash]
$dot -Tjpg automaton.gv -o automaton.jpeg 
\end{lstlisting} 
See \href{http://www.graphviz.org/}{Graphviz} to learn how to convert
.gv files to other file types.

\printbibliography[heading=bibintoc]
\end{document}